\documentclass[letterpaper, 10 pt, conference]{ieeeconf}  

\IEEEoverridecommandlockouts                              

\usepackage{makecell}
\usepackage{amsfonts}
\usepackage{amstext}
\usepackage{mathtools}
\usepackage{cases}
\usepackage{enumerate}
\usepackage{subfigure}
\usepackage{amssymb}
\usepackage{float}
\usepackage{comment}
\usepackage[dvipsnames]{xcolor}
\usepackage{hyperref}
\usepackage{matlab-prettifier}
\usepackage{multirow}

\usepackage{algorithm}
\usepackage{algpseudocode}
\usepackage{etoolbox}



\newtheorem{lemma}{\textbf{Lemma}}
\newtheorem{theorem}{\textbf{Theorem}}
\newtheorem{corollary}{\textbf{Corollary}}
\newtheorem{definition}{\textbf{Definition}}
\newtheorem{remark}{\textbf{Remark}}

\newtheorem{assumption}{\textbf{Assumption}}

\newtheorem{proposition}{\textbf{Proposition}}

\newcommand{\sign}{\mbox{sgn}}

\title{\LARGE \bf Robust Adaptive Learning Control for a Class of Non-affine Nonlinear Systems}
\author{Shuai Gao, Dong Shen, and Abdelhamid Tayebi  
\thanks{This work was supported in part by the National Natural Science Foundation of China under Grant 62573422 and Grant 62173333; and in part by the Outstanding Innovative Talents Cultivation Funded Programs 2025 of Renmin University of China.} 
\thanks{S. Gao, and D. Shen are with the School of Mathematics, Renmin University of China, Beijing 100872, China (e-mail: gaoshuai@ruc.edu.cn, dshen@ieee.org).} 
\thanks{A. Tayebi is with the Department of Electrical and Computer Engineering, Lakehead University, Thunder Bay, ON P7B 5E1, Canada (e-mail: atayebi@Lakeheadu.ca).}
}

\begin{document}

\maketitle
\thispagestyle{empty}
\pagestyle{empty}

\begin{abstract}
We address the tracking problem for a class of uncertain non-affine nonlinear systems with high relative degrees, performing non-repetitive tasks. We propose a rigorously proven, robust adaptive learning control scheme that relies on a gradient descent parameter adaptation law to handle the unknown time-varying parameters of the system, along with a state estimator that estimates the unmeasurable state variables. Furthermore, despite the inherently complex nature of the non-affine system, we provide an explicit iterative computation method to facilitate the implementation of the proposed control scheme. The paper includes a thorough analysis of the performance of the proposed control strategy, and simulation results are presented to demonstrate the effectiveness of the approach.
\end{abstract}

\section{Introduction}

The field of iterative learning control (ILC) has considerably progressed since the early developments in \cite{Arimoto}. ILC techniques focus on utilizing the results of previous operations to modify the control signal for future tasks. This enables the controller to improve its performance gradually with each iteration, ultimately achieving precise tracking after several cycles. Many controlled systems are designed to repeatedly perform similar tasks within a finite time period, such as robot manipulators \cite{TAYEBI2004}, mobile robots \cite{JIN201863}, and high-speed trains \cite{10083251}. For this type of system, ILC is an effective control strategy that leverages input and output data from previous iterations to improve the control performance in subsequent iterations. The reader is referred to the following non-exhaustive list of references in the ILC field \cite{Survey2006,Moore2007,Shen2018,SHEN2019,9773948,8747543,10423244}.

Non-affine nonlinear systems are commonly encountered in practical control scenarios, such as underwater vehicles \cite{GERANMEHR2015248}, flight control systems \cite{Boskovic2004}, and chemical reactions including PH neutralization \cite{Ge2013}. However, there is no systematic control design approach for non-affine nonlinear systems. The complexity in dealing with this type of systems arises from the fact that the inputs influence the system dynamics in a nonlinear manner, making their control a challenging task \cite{Ge2013}. Despite the considerable progress in the ILC field, there are only a few papers in the literature dealing with non-affine nonlinear systems \cite{9184115,9721014,9516984,chi2020JSC,8676086,XU2024123339}.

In \cite{9184115}, a P-type ILC scheme was proposed for non-affine nonlinear systems, and it was shown that, even for systems with relative degree of one, the  selection of the learning gain ensuring the convergence of tracking error is not a trivial task. Some data-driven ILC (DDILC) approaches, based on dynamic linearization, were also developed for non-affine nonlinear systems \cite{9721014,9516984,chi2020JSC}, where the non-affine characteristics are modeled using pseudo-partial derivatives (PPDs). However, dynamic linearization requires the system to satisfy the generalized Lipschitz condition \cite{9721014,9516984,chi2022data,10560025}, which is more stringent than the common global Lipschitz condition and is often difficult to satisfy in practice. Furthermore, the convergence process of the DDILC method is based on the contraction mapping principle, which is strongly related to the task being performed. Once the reference trajectory varies, the contraction process must restart. Therefore, the DDILC approach encounters challenges in effectively tracking iteration-varying reference trajectories. Neural networks (NN) approximation capabilities have also been leveraged to design ILC strategies for the tracking problem of non-repetitive reference trajectories for non-affine nonlinear systems \cite{8676086,XU2024123339}. However, the approximation result of the NN only holds on a local compact set. The issue of ensuring that the system state remains within this compact set has not been adequately addressed. This local approximation issue weakens the theoretical completeness of NN-based ILC. On the other hand, network complexity and approximation accuracy are inherently conflicting goals and cannot be optimized simultaneously, which further limits the practicality of NNs.

The core idea behind the aforementioned methods, consisting in transforming non-affine nonlinear systems into equivalent linear systems via dynamic linearization or using neural networks to approximate the learning controller, is to bypass the complexities of non-affine systems. As such, these approaches are considered indirect control strategies. This raises the question of whether it is possible to develop direct ILC methods for certain non-affine nonlinear systems that can effectively track reference trajectories, which change across iterations.

Real-world non-affine systems often possess a well-defined physical basis ({\textit e.g.}, the aircraft system in \cite{Boskovic2004} and the double inverted pendulum system in \cite{Hovakimyan01012001}). Their dynamic structure can typically be derived from physical laws, while the main source of uncertainty lies in unknown parameters. Such parameterized non-affine systems are both common in engineering practice and structurally exploitable, which is consistent with our intention of designing a direct ILC approach to deal with this type of systems.
For the sake of generality, we accommodate systems with high relative degrees and non-repetitive disturbances. By employing appropriate system identification techniques and treating identification errors as unknown disturbances, our framework can potentially be extended to more general, non-parameterized systems. For parameterized systems, adaptive ILC (AILC), which combines ILC with traditional adaptive control, is a natural and promising approach. However, most existing AILC results have been restricted to affine systems \cite{TAYEBI2004,Tayebi2007,Chien-Tayebi2008,10290922,10324301,Yu2013DiscretetimeAI,7880611,CHI20082207}, leaving the class of parameterized non-affine systems largely unexplored.

In the present work, we propose a direct AILC scheme for a class of uncertain discrete-time parameterized non-affine nonlinear systems, with high relative degrees and non-repetitive disturbances, enabling effective tracking of iteration-varying reference trajectories over finite time-intervals. Compared to the existing works on AILC for affine systems \cite{TAYEBI2004,Tayebi2007,Chien-Tayebi2008,10290922,10324301,Yu2013DiscretetimeAI,7880611,CHI20082207}, the non-affine characteristic presents unique technical challenges for controller design and analysis. First, the feasibility of the trajectory tracking problem in the affine case is determined by the input gain. In the non-affine case, however, it is equivalent to the solvability of a nonlinear equation with respect to the input variable, for which no conclusions are available. Second, the AILC input for affine systems is derived from the certainty equivalence principle and given as an explicit expression. However, in the non-affine case, the AILC input is given in the form of an implicit function manner. Therefore, an efficient numerical implementation scheme for the control input must be designed. Furthermore, analyzing the impact of numerical approximation errors of the control input on the trajectory tracking error is a unique challenge for non-affine systems. Finally, the consideration of practical factors such as non-repetitive disturbances and high relative degrees, while enhancing the generalizability of the control scheme, further complicates these technical challenges. Therefore, the problem considered in this work cannot be solved within the current AILC  framework and requires new technical methods.

We rely on a gradient descent parametric adaptation (GDPA) law along with a state estimator that provides estimates for the unmeasurable state variables, to derive a robust AILC ensuring the convergence of the tracking error to a neighborhood of zero in the presence of external disturbances, and a perfect tracking in the absence of disturbances. The existence and uniqueness of the control input is proven using the global implicit function theorem \cite{ZHANG2006251}.

Since the system at hand is non-affine, an explicit solution for the control input may not always exist. To address this issue, we propose a contraction mapping-based iterative numerical computation method that approximates the control input, facilitating the implementation of our proposed control strategy. A detailed analysis of the impact of the approximation errors on the tracking accuracy is provided.  Two simulation examples are also provided to illustrate the effectiveness of the proposed approach.

The main contributions of this work are outlined below:
\begin{itemize}
\item {We propose an effective control scheme for a class of uncertain non-affine nonlinear systems, with high relative degrees and non-repetitive disturbances. Due to the inherent complexity of this type of systems, studies on their control design remain relatively scarce.}

\item {By analyzing the structural characteristics of non-affine systems, we propose a direct AILC method that is significantly different from the existing DDILC \cite{9721014,9516984,chi2020JSC} and NN-based ILC \cite{8676086,XU2024123339} methods. It avoids the use of neural networks and relaxes the stringent requirements on the system and reference trajectory (for example, the generalized Lipschitz condition and a repetitive reference trajectory). The proposed design path consisting of \textit{unknown parameters estimation}, \textit{implicit control input design} and \textit{explicit numerical approximation of the designed control input}, provides a new idea for solving trajectory tracking problems of non-affine nonlinear systems.}

\item {Different from existing works on AILC \cite{TAYEBI2004,Tayebi2007,Chien-Tayebi2008,10290922,10324301,Yu2013DiscretetimeAI,7880611,CHI20082207}, we develop an AILC method for non-affine systems, conduct the feasibility analysis of the tracking problem, and provide the implicit control input design and its numerical implementation algorithm. The relationship between tracking accuracy and the disturbances, relative degree, and the numerical approximation error of the control input is quantitatively characterized.}
\end{itemize}

The remainder of this paper is organized as follows. Section \ref{PF} formulates the tracking control problem. The feasibility analysis of the tracking control problem and AILC design are provided in Section \ref{Cdesign}. The performance of the proposed AILC scheme is analyzed in Section \ref{BasicTracking}. Section \ref{Iterativemethod} provides an explicit and implementable solution for the proposed AILC scheme. Section \ref{simu} provides some simulation examples. Section \ref{conclu} wraps up this paper.

\emph{Notations}: The set of positive integers is denoted by $\mathbb{Z}^+$. The set of real numbers and the $n$-dimensional real space are denoted by $\mathbb{R}$ and  $\mathbb{R}^n$, respectively. For a vector $\boldsymbol{\alpha}\in \mathbb{R}^n$, $\|\boldsymbol{\alpha}\|$ denotes the 2-norm of $\boldsymbol{\alpha}$. The ball in $\mathbb{R}^n$ with center $\boldsymbol{x_0}$ and radius $r_0$ is denoted by $\mathcal{B}(\boldsymbol{x_0},r_0):=\{\boldsymbol{x} \in \mathbb{R}^n, ~\|\boldsymbol{x}-\boldsymbol{x_0}\|\leq r_0\}$.
\section{Problem Formulation}\label{PF}

Consider the following non-affine nonlinear repetitive system:
\begin{align}\label{system}
    x_k(t+\rho)=F(\boldsymbol{X}_k(t),u_k(t),t)+w_k(t),
\end{align}
where $\boldsymbol{X}_k(t)=[x_k(t+\rho-1),x_k(t+\rho-2),\cdots,x_k(t)]^\top$. The variables $x_k(t)\in \mathbb{R}$, $u_k(t)\in \mathbb{R}$, and $w_k(t)\in \mathbb{R}$ denote the measurable state, system input, and unknown disturbance, respectively. The  time variable is denoted by $t\in\{0,1,\cdots,T-\rho\}$ and the iteration variable is denoted by $k\in \mathbb{Z}^+$. The nonlinear function $F(\boldsymbol{X}_k(t),u_k(t),t)$ satisfies the following:
\begin{align*}
    F(\boldsymbol{X}_k(t),u_k(t),t)&=\boldsymbol\theta(t)^\top \boldsymbol{f}(\boldsymbol{X}_k(t),u_k(t)),
\end{align*}
where $\boldsymbol\theta(t)\in \mathbb{R}^{p}$ is an unknown time-varying parameter vector, and $\boldsymbol{f}(\boldsymbol X,u): \mathbb{R}^{\rho+1} \to \mathbb{R}^p$ is a known function vector with sufficiently smooth components. According to the definition in \cite{802910}, the concept of relative degree denotes the delay between the input and the output. For system \eqref{system}, the control input at time $t$ determines the output $\rho$ steps later. Therefore, the relative degree of \eqref{system} is $\rho$.

For our control design, we need the following assumptions:
\begin{assumption}\label{assu1}
 There exists a known closed neighborhood $\mathcal{B}(\underline{\boldsymbol\theta},R)$ such that
        $\boldsymbol\theta(t)\in \mathcal{B}(\underline{\boldsymbol\theta},R)$, and
        for any $p$-dimensional vector $\boldsymbol\phi \in \mathcal{B}(\underline{\boldsymbol\theta},R)$, $\boldsymbol{X}\in \mathbb{R}^{\rho}$, $u\in\mathbb{R}$, the inequality
$\left| \boldsymbol\phi^\top\frac{\partial}{\partial u} \boldsymbol{f}(\boldsymbol{X},u)\right|>d_0$ holds, where $d_0$ is an unknown positive constant.
\end{assumption}
\begin{assumption}\label{assu2}
The system disturbance $w_k(t)$ is bounded, {\textit i.e.,}  $w\triangleq\sup_{k,t}|w_k(t)|<\infty$,
where $w$ is an unknown supremum.
\end{assumption}
\begin{assumption}\label{assu3}
The vector-valued function $\boldsymbol{f}(\boldsymbol{X},u)$ satisfies the global Lipschitz condition. In other words, there exist positive constants $L^{\boldsymbol{X}}$ and $L^{u}$ such that
\begin{align*}
&\|\boldsymbol{f}(\boldsymbol{X}_1,u)-\boldsymbol{f}(\boldsymbol{X}_2,u)\|\leq L^{\boldsymbol{X}}\|\boldsymbol{X}_1-\boldsymbol{X}_2\|,\ \\
&\|\boldsymbol{f}(\boldsymbol{X},u_1)-\boldsymbol{f}(\boldsymbol{X},u_2)\|\leq L^u|u_1-u_2|.
\end{align*}
\end{assumption}
\begin{assumption}\label{assu4}
 The variables $x_k(0)$, $x_k(1)$,..., $x_k(\rho-1)$ are bounded and known for all $k\in \mathbb{Z}^+$.
\end{assumption}
\begin{remark}\label{remark1}
 Assumption \ref{assu1} is equivalent to requiring a non-vanishing input gain for affine systems as in \cite{CHI20082207,Yu2013DiscretetimeAI} and is often assumed for non-affine systems \cite{8957439,8676086}. While this study focuses on single-input, single-output (SISO) systems, the proposed AILC framework is expected to generalize to multiple-input, multiple-output (MIMO) scenarios. The key is to identify reasonable assumptions that guarantee tracking feasibility (similar to Assumption 1 in the SISO case). For example, if the Jacobian matrix of the system function with respect to the input vector is strictly diagonally dominant, then the generalization can be achieved. More general feasibility conditions need further detailed exploration. Since most disturbances in real-world systems are bounded, Assumption \ref{assu2} is realistic and commonly adopted \cite{9321505,10750414,9656625}. Moreover, the results of this study can be extended to certain specific cases of unbounded disturbances, such as parameterized disturbances. Assumption \ref{assu3} imposes some restrictions on the nonlinearity as in \cite{9321505,1266787,4026648}, which is satisfied for a large class of practical non-affine systems, {\textit e.g.,} complex circuits \cite{6858093} and single-link manipulators \cite{6389712}. Given the system's relative degree $\rho$, the initial $\rho$ values of each iteration, $x_k(0)$, $x_k(1)$,..., $x_k(\rho-1)$, cannot be controlled and are determined by resetting the system before the start of iteration $k$. Therefore, we assume them to be known and bounded.
\end{remark}
\begin{remark}\label{remark2}
The following strict feedback system is commonly found in the existing literature \cite{10542381,10485507,CHI20082207}:
\begin{align}\label{strictfeedback}
\left\{ \begin{array}{l}
	x_k^{(j)}(t+1)=x_k^{(j+1)}(t)+f_{j}\big(x_k^{(1)}(t),\cdots,x_k^{(j)}(t)\big),  \\
 \hspace{50mm} j=1\cdots \rho-1, \\
        x_k^{(\rho)}(t+1)=\boldsymbol\theta(t)^\top f_{\rho}(x_k^{(1)}(t),\cdots,x_k^{(\rho)}(t),u_k(t)),
	\end{array}\right.
\end{align}
where $f_j,\ j=1,\cdots,\rho$ are known functions satisfying the global Lipschitz condition, and $\boldsymbol\theta(t)$ is an unknown time-varying parameter vector. Consider $x_k^{(1)}(t)$ as a measurable state and rewrite $x_k^{(i)}(t)$ as a function of $x_k^{(1)}(t),x_k^{(1)}(t+1),\cdots,x_k^{(1)}(t+i-1)$, $i=1,\cdots,\rho$. Algebraic recursions verify that system \eqref{strictfeedback} can be transformed into the form of system \eqref{system}.
\end{remark}

\emph{Objective}: Under Assumptions \ref{assu1}-\ref{assu4}, design an AILC input sequence $\{u_k(t)\}_{k=1}^\infty$ such that the state $\{x_k(t)\}_{k=1}^\infty$ tracks the iteration-varying reference trajectory $\{r_k(t)\}_{k=1}^\infty$ as accurately as possible for all  $t\in \{\rho,\cdots, T\}$ when $k$ tends to infinity. Denoting the tracking error as $e_k(t)=x_k(t)-r_k(t)$, our AILC should guarantee
\begin{align*}
\limsup_{k\to \infty} |e_k(t)|\in \mathcal{B}(0,\mathcal{R}(w)), \ t\in \{\rho,\cdots, T \},
\end{align*}
where $\mathcal{R}(w)$ is the radius of the convergence ball depending on the disturbance supremum, and
\begin{align*}
\lim_{k\to \infty} |e_k(t)|=0,\  t\in \{\rho,\cdots, T\},
\end{align*}
in the absence of disturbances (\textit{i.e.,} $w_k(t)\equiv 0$).

To achieve the control objective, we must design a parametric adaptation law to deal with the unknown parameter vector $\boldsymbol\theta(t)$ in the control scheme and a state estimator that provides estimates of the unmeasurable state variables. Since an explicit solution for the AILC input may not always exist, an efficient numerical method will be proposed to design an explicitly implementable control input.

\section{Tracking Feasibility and AILC Design}\label{Cdesign}

\subsection{Tracking Feasibility Analysis}\label{Trackingability}

Let us first focus on the feasibility of the tracking control problem. We will show the existence of an ideal control input $u_k^\ast(t)$ such that $x_k(t+\rho) = r_k(t+\rho)$ for all  $t\in \{0,\cdots, T-\rho \}$ for system \eqref{system} with $w_k(t)\equiv 0$. To this end, it is sufficient to check the solvability of the following equation with respect to $u$:
\begin{align} \label{trackingequation}
      \boldsymbol{\theta}(t)^\top \boldsymbol{f}(\boldsymbol{X}_k(t),u)-r_k(t+\rho)=0.
\end{align}

As per Assumption \ref{assu1}, one has $|\boldsymbol\theta(t)^\top\frac{\partial }{\partial u}\boldsymbol{f}(\boldsymbol{X}_k(t),u)|> d_0>0$. For the sake of simplicity and without loss of generality, let us consider the case $\boldsymbol\theta(t)^\top\frac{\partial }{\partial u}\boldsymbol{f}(\boldsymbol{X}_k(t),u)> d_0$, and define the mapping \(\mathcal{Z}:\mathbb{R}\to\mathbb{R}\), \(u\mapsto \mathcal{Z}(u)=\boldsymbol\theta(t)^\top \boldsymbol{f}(\boldsymbol{X}_k(t),u)-r_k(t+\rho)\).
Let $c=\mathcal{Z}(0)$. If $c=0$, then  $\boldsymbol{\theta}(t)^\top \boldsymbol{f}(\boldsymbol{X}_k(t),0)=r_k(t+\rho)$, which means $u=0$ already satisfies \eqref{trackingequation}. If $c\neq 0$, we define the following mapping: \(\mathcal{T}:\ \mathcal{B}\Big(0,\frac{|c|}{d_0}\Big)\to\mathbb{R},~ u\mapsto\mathcal{T}(u)=u-\frac{1}{l}\mathcal{Z}(u)\), where $l>0$ is a constant such that $\mathcal{T}$ is a contraction mapping on $\mathcal{B}\Big(0,\frac{|c|}{d_0}\Big)$. Then $\mathcal{T}$ has a unique fixed point, which is the solution of \eqref{trackingequation} as stated in the following proposition:
\begin{proposition}\label{pp1}
    Under Assumption \ref{assu1}, for any $t\in \{0,1,\cdots$, $T-\rho\}$ and $k\in \mathbb{Z}^+$, there exists a unique ideal control input $u=u_k^\ast(t)$ such that \eqref{trackingequation} holds. In addition, there exists a known contraction mapping $\mathcal{T}$ such that $u_k^\ast(t)$ is the fixed point of $\mathcal{T}$, {\textit i.e.,} $\mathcal{T}(u_k^\ast(t))=u_k^\ast(t)$.
\end{proposition}

\emph{Proof:~} See Appendix A.

\subsection{GDPA Scheme}\label{GDPAandObserver}

Proposition \ref{pp1} shows that, if $w_k(t)\equiv 0$, there exists an ideal control input $u_k^\ast(t)$ leading to a perfect tracking $x_k(t+\rho)=r_k(t+\rho)$. However, $u_k^\ast(t)$ depends on $\boldsymbol\theta(t)$, which is unknown, and $\boldsymbol{X}_k(t)$, which contains unmeasurable components $x_k(t+1)$, $x_k(t+\rho-2)$,..., $x_k(t+\rho-1)$. Therefore, we must estimate $\boldsymbol\theta(t)$ and $\boldsymbol{X}_k(t)$. In the sequel, we will denote the estimate of $\boldsymbol\theta(t)$ at the $k$-th iteration by $\hat{\boldsymbol\theta}_k(t)$. The adaptation law of $\hat{\boldsymbol\theta}_k(t)$ is obtained by minimizing the following cost function:
\begin{align*}
    J(\hat{\boldsymbol\theta}(t))=\frac{(x_k(t+\rho)-\hat{\boldsymbol\theta}(t)^\top \boldsymbol{f}(\boldsymbol{X}_k(t),u_k(t)))^2}{2m_k(t)^2},\ \textbf{w.r.t}\ \hat{\boldsymbol\theta}(t)
\end{align*}
where $m_k(t)=\sqrt{1+\|\boldsymbol{f}(\boldsymbol{X}_k(t),u_k(t))\|^2}$ is a normalized factor. Using the gradient descent method, one can derive the parametric adaptation law in the following form:
\begin{align*}
\hat{\boldsymbol\theta}_{k+1}(t)&=\hat{\boldsymbol\theta}_k(t)+\eta \frac{(x_k(t+\rho)-\hat{\boldsymbol\theta}_k(t)^\top \boldsymbol{f}(\boldsymbol{X}_k(t),u_k(t)))}{m_k(t)^2}\\
&\quad\times \boldsymbol{f}(\boldsymbol{X}_k(t),u_k(t)),
\end{align*}
where $\eta$ is an adaptation gain.

Due to system disturbances, robustness modifications are necessary. We introduce an iteration-varying dead-zone function $a_k(t)$ and a projection operator $\mathrm{Proj}(\cdot)$ to design the final GDPA law as follows:
\begin{align}
 &\epsilon_k(t+\rho)=\frac{x_k(t+\rho)-\hat{\boldsymbol\theta}_k(t)^{\top}\boldsymbol{f}(\boldsymbol{X}_k(t),u_k(t))}{m_k(t)^2},\label{GDPA1} \\
   &\Bar{\boldsymbol\theta}_{k+1}(t)=\hat{\boldsymbol\theta}_{k}(t)+\eta a_k(t)\epsilon_k(t+\rho)\boldsymbol{f}(\boldsymbol{X}_k(t),u_k(t)),\label{GDPA2}\\
   &a_k(t)=\left\{ \begin{array}{l}
	 0,\ \text{if}\  |\epsilon_k(t+\rho)|\leq \frac{\hat{w}_{k}(t)}{m_k(t)^2},\\
 1-\frac{\hat{w}_k(t)}{|\epsilon_k(t+\rho)|m_k(t)^2}, \ \text{otherwise},
	\end{array}\right. \label{GDPA3}\\
 & \hat{w}_{k+1}(t)=\hat{w}_{k}(t)+\eta a_k(t)|\epsilon_k(t+\rho)|, \label{GDPA4}\\
   &\hat{\boldsymbol\theta}_{k+1} (t)=\mathrm{Proj}_{\mathcal{B}(\underline{\boldsymbol\theta},R)}\left(\Bar{\boldsymbol\theta}_{k+1} (t)\right)\nonumber  \\
   &\qquad\quad\:=\left\{ \begin{array}{l}
   \Bar{\boldsymbol\theta}_{k+1} (t),\  \text{if}\  \Bar{\boldsymbol\theta}_{k+1} (t)\in \mathcal{B}(\underline{\boldsymbol\theta},R), \\
	\underline{\boldsymbol\theta}+\frac{R(\Bar{\boldsymbol\theta}_{k+1} (t)-\underline{\boldsymbol\theta})}{\|\Bar{\boldsymbol\theta}_{k+1} (t)-\underline{\boldsymbol\theta}\|},\ \text{otherwise}.\end{array}\right.  \label{GDPA5}
\end{align}
The initial value $\hat{\boldsymbol\theta}_0(t)$ is chosen arbitrarily. The variable $\epsilon_k(t+\rho)$ represents the normalized modeling error with factor $m_k(t)$. The dead-zone signal $a_k(t)$ compensates for system disturbances. Given the unknown disturbance supremum $w$, we use the estimated supremum $\hat{w}_{k}(t)$, initializing $\hat{w}_{0}(t)=0$ to ensure $\hat{w}_{k}(t)\geq 0$. The projection operator $\mathrm{Proj}_{\mathcal{B}(\underline{\boldsymbol\theta},R)}(\cdot)$ restricts $\hat{\boldsymbol\theta}_k(t)$ within the closed ball $\mathcal{B}(\underline{\boldsymbol\theta},R)$. The estimates $\hat{\boldsymbol\theta}_k(t)$ and $\hat{w}_k(t)$ are updated iteratively using data from previous trials, ensuring that the GDPA law \eqref{GDPA1}-\eqref{GDPA5} adheres to causality.

\subsection{State Estimation}

 Applying the GDPA law, at time $t$ of the $k$-th iteration, the state variables $x_k(0),\cdots,x_k(t)$, inputs $u_k(0),\cdots,u_k(t-1)$, and the parameter estimates $\hat{\boldsymbol\theta}_k(0),\cdots,\hat{\boldsymbol\theta}_k(T-\rho)$ are available. This information is then used to design the state estimator to estimate the unmeasurable state variables in $\boldsymbol{X}_k(t)$ as follows:

{\small
\begin{align}
x^e_k(t+\rho-1|t)&=\hat{\boldsymbol\theta}_k(t-1)^\top \boldsymbol{f}\big(x^e_k(t+\rho-2|t),\cdots, \label{estimateX1}\\
~&x^e_k(t+1|t),x_k(t),x_k(t-1),u_k(t-1)\big) \nonumber \\
x^e_k(t+\rho-2|t)&=\hat{\boldsymbol\theta}_k(t-2)^\top \boldsymbol{f}\big(x^e_k(t+\rho-3|t),\cdots,\\
~&x^e_k(t+1|t),x_k(t),x_k(t-1),x_k(t-2),u_k(t-2)\big) \nonumber \\
&~~\vdots \nonumber\\ \nonumber
x_k^e(t+1|t)&=\hat{\boldsymbol\theta}_k(t-\rho+1)^\top \boldsymbol{f}\big(x_k(t),\cdots,x_k(t-\rho+1),\\
~&u_k(t-\rho+1)\big), \label{estimateX2} \\
x_k^e(t|t)&=x_k(t), \label{estimateX3}
\end{align}
}For $0 \leq t'\leq  t \leq \rho-1$, we set $x_k^e(t|t')=x_k(t)$ as $x_k(0)$, $\ldots$, $x_k(\rho-1)$ are available (recall Assumption \ref{assu4}). The estimated state at time $t+\rho$ is obtained by the following model:
\begin{align}\label{predictionmodel}
x^e_k(t+\rho|t)&=\hat{\boldsymbol\theta}_k(t)^\top \boldsymbol{f}(\boldsymbol{X}_k^e(t),u_k(t)),
\end{align}
where $\boldsymbol{X}_k^e(t)=[x^e_k(t+\rho-1|t),\cdots,x_k^e(t+1|t),x_k^e(t|t)]^\top$.
\subsection{Implicit AILC Design}\label{Predictionandinput}
 Now we will use model \eqref{predictionmodel} with the available data to show the existence of an implicit AILC input $u_k(t)$. Letting $x^e_k(t+\rho|t) = r_k(t+\rho)$ and considering the model \eqref{predictionmodel}, the implicit AILC input $u_k(t)$ is the solution of the following equation:
 \begin{align}\label{pred-model}
 \hat{\boldsymbol\theta}_k(t)^\top \boldsymbol{f}(\boldsymbol{X}_k^e(t),u_k(t))= r_k(t+\rho).
\end{align}

The existence of the AILC input $u_k(t)$, as a solution of \eqref{pred-model}, is stated in the following proposition:

\begin{proposition}\label{pp2}
Under Assumption \ref{assu1}, there exists a unique AILC input $u_k(t)$ such that \eqref{pred-model} holds. In addition, there exists a  known contraction mapping  $\mathcal{T}':\mathcal{B}(0,\frac{|c'|}{d_0})\to\mathbb{R}$, $\mathcal{T}'(u)=u-\frac{1}{l'}\mathcal{Z}'(u)$ such that $u:=u_k(t)$ is the fixed point of $\mathcal{T}'$, where $\mathcal{Z}'(u)=\hat{\boldsymbol\theta}_k(t)^\top \boldsymbol{f}(\boldsymbol{X}_k^e(t),u)-r_k(t+\rho)$, $c'=\mathcal{Z}'(0)$ and  $l'>\sup_{u\in \mathcal{B}(0,\frac{|c'|}{d_0})}|\hat{\boldsymbol\theta}_k(t)^\top\frac{\partial}{\partial u} \boldsymbol{f}(\boldsymbol{X}_k^e(t),u)|$.
\end{proposition}

\emph{Proof:~} The proof is similar to that of Proposition \ref{pp1} and thus omitted for brevity.

Given the parameter estimation $\hat{\boldsymbol\theta}_k(t)$, the estimated state vector $\boldsymbol{X}_k^e(t)$, and the reference trajectory $r_k(t+\rho)$, the AILC input $u_k(t)$ is an implicit function given by
\begin{align}\label{ALCLaw}
    u_k(t)=\mathcal{G}(\boldsymbol{X}_k^e(t),\hat{\boldsymbol\theta}_k(t),r_k(t+\rho)).
\end{align}

The implicit AILC scheme, which includes the GDPA law \eqref{GDPA1}-\eqref{GDPA5}, state estimator  \eqref{estimateX1}-\eqref{estimateX3}, and control input \eqref{ALCLaw} is illustrated in Fig.\,\ref{fig:block}.

\begin{remark}\label{remark3}
It is worth pointing out that he GDPA law \eqref{GDPA1}-\eqref{GDPA5} can be simplified for specific scenarios. If the function $\boldsymbol{f}(\boldsymbol X,u)$ is bounded, we can set $m_k(t) \equiv 1$. In the disturbance-free case, we can set $a_k(t) \equiv 1$ and omit the estimation $\hat{w}_k(t)$, to improve the tracking performance (see Corollary \ref{corollarydisturbance} in Section \ref{BasicTracking} and Corollary \ref{corollarynodis+appro} in Section \ref{Iterativemethod}). If disturbance $w_k(t)$ has a known upper bound $w^+$ ($|w_k(t)|\leq w^+$), we can use $w^+$ in \eqref{GDPA3} and skip the estimation law \eqref{GDPA4}, while a large difference between $w^+$ and $w$ may affect the tracking accuracy.
\end{remark}

\begin{figure}[!t]
	\centering
	\includegraphics[width=0.50\textwidth]{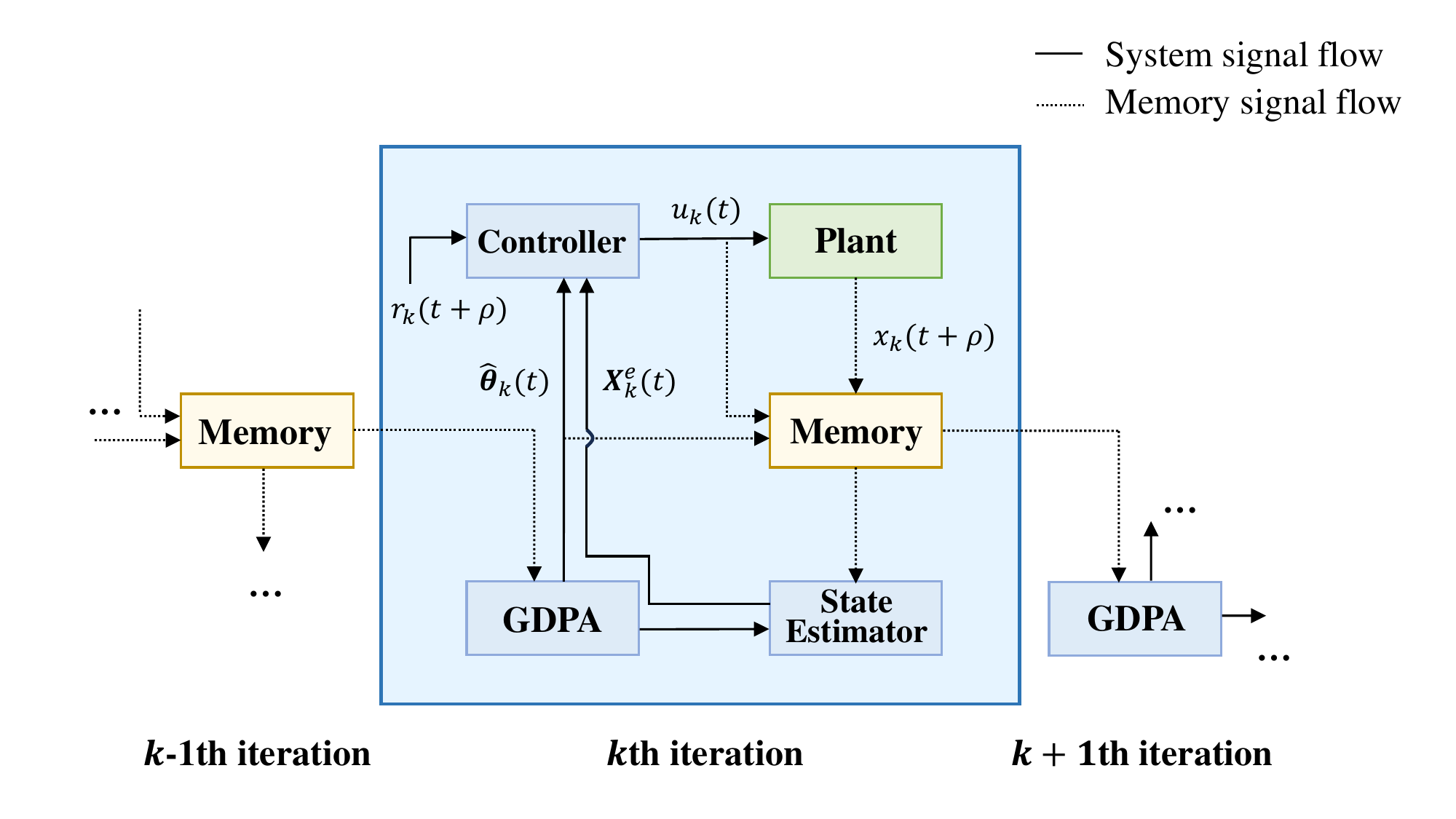}
	\caption{Block diagram of the AILC scheme.}
	\label{fig:block}
\end{figure}

\section{Convergence Analysis}\label{BasicTracking}

In this section, we will establish the convergence properties of the proposed AILC scheme \eqref{ALCLaw} with the GDPA law \eqref{GDPA1}-\eqref{GDPA5} and the state estimates \eqref{estimateX1}-\eqref{estimateX3}. Our result is stated in the following theorem:

\begin{theorem}\label{Th1}
Consider system \eqref{system} satisfying Assumptions \ref{assu1}-\ref{assu4}, under the AILC input \eqref{ALCLaw} with the GDPA law \eqref{GDPA1}-\eqref{GDPA5} and the state estimates \eqref{estimateX1}-\eqref{estimateX3}, with $0<\eta <2$, for all $t\in \{0,\cdots,T-\rho\}$, there exist constants $K^{i,t}\geq0,\ i=0,\cdots,t-1$, such that
\begin{equation*}
    \limsup_{k\to \infty}|e_k(t+\rho)|\leq w+\sqrt{2V_0(t)}+\sum_{i=0}^{t-1} K^{i,t}\left(w+\sqrt{2V_0(i)}\right),
\end{equation*}
where $V_0(t)=\frac{1}{2}\|\hat{\boldsymbol\theta}_0(t)-\boldsymbol\theta(t)\|^2+\frac{1}{2}(\hat{w}_{0} (t)-w)^2$.
\end{theorem}

\emph{Proof:~} See Appendix B.

Theorem \ref{Th1} shows that the tracking error $e_k(t+\rho)$ converges to a closed ball around zero. On the one hand, the radius of the closed ball is related to $K^{i,t}$, which depends on the unknown constant $d_0$, Lipschitz constants $L^{\boldsymbol{X}}$, $L^u$, upper bound of the reference trajectory $\sup_{k,t}r_k(t+\rho)$, initial estimation $\hat{\boldsymbol{\theta}}_0(t)$, and radius $R$. These values cannot be improved no matter what controller is designed. Therefore, the upper bound in Theorem \ref{Th1} is tight. On the other hand, the radius depends on the disturbances from time $0$ to $t$. This is due to the relative degree $\rho$, which results in unmeasurable state variables $x_k(t+\rho-1)$,..., $x_k(t+1)$. Therefore, we use historical information prior to time $t$ to estimate the unmeasurable state variables (see \eqref{estimateX1}-\eqref{estimateX3}). For the special case of $\rho=1$, $u_k(t)$ can be obtained from the model $r_k(t+1)=\hat{\boldsymbol\theta}_k(t)^\top \boldsymbol{f}(x_k(t),u_k(t))$, eliminating the need for state estimations, leading to the following corollary:
\begin{corollary}\label{corollaryrho=1}
For $\rho=1$, under the conditions of Theorem \ref{Th1}, for each $t\in\{0,\cdots,T-1\}$, one has
\begin{align*}
   \limsup_{k\to \infty} |e_k(t+1)|\leq w+\sqrt{2V_0(t)}.
\end{align*}
\end{corollary}

\emph{Proof:~} See Appendix C.

In the disturbance-free case, \textit{i.e.,} $w_k(t)\equiv0$, we set the dead-zone function $a_k(t)$ to one, which simplifies the parametric adaptation law \eqref{GDPA2} to:
\begin{align}\label{GDPAmodify}
    \Bar{\boldsymbol\theta}_{k+1}(t)=\hat{\boldsymbol\theta}_{k}(t)+\eta \epsilon_k(t+\rho)\boldsymbol{f}(\boldsymbol{X}_k(t),u_k(t)),
\end{align}
leading to a perfect tracking as stated in the following corollary:

\begin{corollary}\label{corollarydisturbance}
 Consider system \eqref{system}  with $w_k(t)\equiv0$, satisfying Assumptions \ref{assu1}, \ref{assu3}, \ref{assu4}, under the AILC input \eqref{ALCLaw} with the GDPA law \eqref{GDPA1}, \eqref{GDPAmodify}, \eqref{GDPA5}, and
the state estimates \eqref{estimateX1}-\eqref{estimateX3}, with $0<\eta<2$, then a perfect tracking is achieved, \textit{i.e.,}
\begin{align*}
    \lim_{k\to \infty} |e_k(t+\rho)|=0,
\end{align*}
for all $t\in \{0,\cdots,T-\rho\}$.
\end{corollary}

\emph{Proof:~} See Appendix D.

\section{Explicit AILC Design}\label{Iterativemethod}

In this section, we provide an explicit and implementable solution for the AILC input \eqref{ALCLaw}. To this end, we propose an iterative computation method that construct a sequence $\{\mathrm{u}_k^{p}(t)\}_{p=0}^{\infty}$, at time $t$ during the $k$-th iteration, as follows:
\begin{align}\label{iterativesolution}
   \mathrm{u}_k^{p}(t)=\left\{ \begin{array}{l}
	\mathcal{T}'(\mathrm{u}_k^{p-1}(t)),\ \ p>0,   \\
         0, \quad  \qquad \qquad \ p=0,
	\end{array}\right.
\end{align}
where the mapping $\mathcal{T}'$ is defined in Proposition \ref{pp2}. One can show that the sequence $\{\mathrm{u}_k^{p}(t)\}_{p=0}^{\infty}$ is convergent due to the fact that $\mathcal{T}'$ is a contraction mapping on $\mathcal{B}(0,\frac{|c'|}{d_0})$.
Letting $p\to \infty$ on both sides of \eqref{iterativesolution}, one obtains
\begin{align*}
\lim_{p\to \infty}\mathrm{u}_k^{p}(t)=\mathcal{T}'\big(\lim_{p\to \infty} \mathrm{u}_k^{p}(t)\big),
\end{align*}
and by the uniqueness of fixed point, one has
\begin{align*}
	\lim_{p\to \infty}\mathrm{u}_k^{p}(t)=u_k(t),
\end{align*}
where, $u_k(t)$ is the implicit control input given by \eqref{ALCLaw}. Note that $\mathrm{u}_k^{p}(t)$ is a numerical approximation of $u_k(t)$ which is accurate at sufficiently large $p$. For practical implementations, one must stop the iterative process $\mathrm{u}_k^{p}(t)=\mathcal{T}'(\mathrm{u}_k^{p-1}(t))$ after a finite number of steps, leading to an approximation error between $\mathrm{u}_k^{p}(t)$ and $u_k(t)$. A practical stopping criterion for this iterative process is provided as follows:

\emph{Iteration Stopping Criterion}: For any given $\varepsilon>0$, the iterative process $\mathrm{u}^{p}_k(t)=\mathcal{T}'(\mathrm{u}_k^{p-1}(t))$ stops after $p_o(k,t)$ iterations, where
\begin{align}\label{p0}
   p_o(k,t)=\hspace{-1mm} \begin{cases}
\big\lfloor \log_{(1-\frac{d_0}{l'})}\frac{\varepsilon d_0}{l'|\mathrm{u}_k^{1}(t)-\mathrm{u}_k^0(t)|} \big\rfloor+1,\ \text{if}\:  \mathrm{u}_k^{1}(t) \neq \mathrm{u}_k^{0}(t), \vspace{1mm}\\
 1, \quad \text{if}\   \mathrm{u}_k^{1}(t)=\mathrm{u}_k^{0}(t).
    \end{cases}
\end{align}
The operator $\lfloor\cdot\rfloor$ denotes the floor function, which rounds a number down to the nearest integer.

Lemma \ref{lemma1} provides the approximation error between the iterative approximation $\mathrm{u}^{p}_k(t)$ and the AILC input $u_k(t)$ under the above iteration stopping criterion.
\begin{lemma} \label{lemma1}
 For any given $\varepsilon>0$, the iterative process \eqref{iterativesolution} guarantees that $|\mathrm{u}^{p_o(k,t)}_k(t)-u_k(t)|<\varepsilon$, where $p_o(k,t)$ is given by the stopping criterion \eqref{p0}.
\end{lemma}

\emph{Proof:~}See Appendix E.
\begin{remark}\label{remark4}
Note that $p_o(k,t)$ is an upper bound of the smallest number of iterations such that $|\mathrm{u}^{p}_k(t)-u_k(t)|<\varepsilon$. Lower values of $\varepsilon$ result in higher values of $p_o(k,t)$ since $\big\lfloor \log_{(1-\frac{d_0}{l'})}\frac{\varepsilon d_0}{l'|\mathrm{u}_k^{1}(t)-\mathrm{u}_k^0(t)|} \big\rfloor$ is a decreasing function of $\varepsilon$. If we consider an addition, a multiplication, an assignment and a comparison of numbers to be one basic operation, and assume $\mathrm{u}_k^{1}(t)-\mathrm{u}_k^0(t)$ is a constant, in the case of $\varepsilon=\frac{1}{n}$, the number of basic operations required by the AILC scheme with numerical approximation \eqref{GDPA1}-\eqref{GDPA5}, \eqref{estimateX1}-\eqref{estimateX3}, \eqref{ALCLaw}, \eqref{iterativesolution} in one iteration ($k$) is $O\Big((\ln n+p\rho)\Tilde{T}\Big)$, where $\Tilde{T}=T-\rho+1$ is the tracking time length. Therefore, as the threshold $\frac{1}{n}$ decreases and the parameter dimension $p$, relative degree $\rho$, tracking time length $\Tilde{T}$ increase, the required computational effort increases.
\end{remark}

Under the AILC input \eqref{ALCLaw} with the GDPA law \eqref{GDPA1}-\eqref{GDPA5}, state estimates \eqref{estimateX1}-\eqref{estimateX3}, and iterative approximation \eqref{iterativesolution}, one can state the following theorem:

\begin{theorem}\label{Th2}
 Consider system \eqref{system} satisfying Assumptions \ref{assu1}-\ref{assu4}, under the AILC input \eqref{ALCLaw} with the GDPA law \eqref{GDPA1}-\eqref{GDPA5}, state estimates \eqref{estimateX1}-\eqref{estimateX3}, and iterative approximation \eqref{iterativesolution}. For any given $\varepsilon>0$, if $p>p_o(k,t)$ and $0<\eta <2$, then there exist constants $K_\ast^{i,t},\ t=0,\cdots,T-\rho,\ i=0,\cdots,t-1$, such that
\begin{align*}
& \limsup_{k\to \infty} |e_k(t+\rho)| \leq w+\sqrt{2V_0(t)}\\
&+\sum_{i=0}^{t-1} K_\ast^{i,t}(w+\sqrt{2V_0(i)})+L^u \left(\sqrt{2V_0(t)}+\|\boldsymbol\theta(t)\|\right) \varepsilon.\\
\end{align*}
\end{theorem}

\emph{Proof:~} See Appendix F.

As stated in Theorem \ref{Th2}, the tracking error bound depends on two sources of uncertainty. The first source is the system disturbances before time $t$, represented by $\sum_{i=0}^{t-1} K_\ast^{i,t}(w+\sqrt{2V_0(i)})$. The second source is the approximation error of the AILC input, given by $L^u \left(\sqrt{2V_0(t)}+\|\boldsymbol\theta(t)\|\right)\varepsilon$. To improve tracking accuracy, one has to reduce $\varepsilon$ by increasing the number of iterations in the iterative process \eqref{iterativesolution}. Similar to Corollary \ref{corollarydisturbance}, in the disturbance-free case, the tracking performance can be improved by using the AILC scheme \eqref{ALCLaw} with the modified GDPA law \eqref{GDPA1}, \eqref{GDPAmodify}, \eqref{GDPA5}, state estimates \eqref{estimateX1}-\eqref{estimateX3}, and iterative approximation \eqref{iterativesolution}.

\begin{corollary}\label{corollarynodis+appro}
Consider system \eqref{system} with $w_k(t)\equiv0$, satisfying Assumptions \ref{assu1}, \ref{assu3}, \ref{assu4}, under the AILC input \eqref{ALCLaw} with the GDPA law \eqref{GDPA1}, \eqref{GDPAmodify}, \eqref{GDPA5}, state estimates \eqref{estimateX1}-\eqref{estimateX3}, and iterative approximation \eqref{iterativesolution}, with $p>p_o(k,t)$ and $0<\eta <2$, then
\begin{align*}
   & \limsup_{k\to \infty} |e_k(t+\rho)|\leq L^u
    \left(\sqrt{2V_0(t)}+\|\boldsymbol\theta(t)\|\right) \varepsilon,\nonumber
\end{align*}
for all $t\in\{0,1,\cdots,T-\rho\}$.
\end{corollary}

\emph{Proof~}: See Appendix G.

\begin{remark}\label{remark5}
Three control input symbols are used so far: $u_k^\ast(t)$, $u_k(t)$, and $\mathrm{u}_k^p(t)$. The variable $u_k^\ast(t)$ represents the ideal control input introduced to establish the result in Proposition \ref{pp1}, showing the feasibility of the tracking problem. This input is not implementable due to the unknown parameters and unmeasurable state variables. The variable $u_k(t)$ represents the implicit control AILC given in \eqref{ALCLaw} as a solution to \eqref{pred-model} and used in Proposition \ref{pp2}. This input depends on the measured and estimated signals, but it is not always implementable since it is given in an implicit manner. The control input $\mathrm{u}_k^p(t)$ is an explicit iterative approximation of $u_k(t)$ that can be used for implementation in the case where an explicit solution to \eqref{pred-model} cannot be obtained.
\end{remark}

\begin{remark}\label{remark6}
According to \cite{9184115}, even for non-affine systems with relative degree of one, a P-type ILC requires careful selection of the learning gain and minimal system disturbances (Assumption A4 in \cite{9184115}). In contrast, our proposed AILC scheme deals with non-affine systems with high relative degrees, and the convergence condition is easier to satisfy. For DDILC, the core of its convergence analysis is the contraction mapping principle. This convergence process is related to the task being performed: Any change in the reference trajectory necessitates restarting the contraction process. Therefore, existing DDILC studies typically assume iteration-invariant reference trajectories \cite{9721014,9516984,chi2020JSC}. The convergence analysis of the proposed AILC approach is based on the Lyapunov-like function techniques and key technical lemma (see Appendix), which is independent of the task being performed. Thus, it can effectively track iteration-varying reference trajectories.
\end{remark}

\section{Numerical Simulations}\label{simu}

In this section, we present two simulation examples to illustrate the effectiveness of the proposed AILC scheme. First, we compare our proposed AILC with the DDILC approach on a non-affine nonlinear system with relative degree of one. Next, we apply the proposed AILC scheme to a double inverted pendulum system with relative degree of two. In these simulations, the AILC input $u_k(t)$ is solved directly from \eqref{pred-model} using ``{\tt solve}'' function in MATLAB.

\subsection{Example 1}\label{simuex1}

The system model is modified from \cite{chi2022data} (Chapter 2, pp. 25) and given as follows:
{\small\begin{align*}
x_k(t+1)&=\left(0.5+\frac{t}{50}\right)\frac{x_k(t)\sin(x_k(t))}{1+x_k(t)^2}+\left(0.75+\frac{t}{75}\right)\mathrm{e}^{\frac{x_k(t)}{100}}\\
&\quad+\left(1.5+0.5\times(-1)^t\right)u_k(t)^3\nonumber \\
&\quad +\sin\left(\frac{\pi}{4}+\frac{\pi}{100}t\right)(\arctan(u_k(t))+u_k(t)),\nonumber \\
&=\boldsymbol\theta(t)^{\top}\boldsymbol{f}(x_k(t),u_k(t)),\qquad  t\in [1,50],
\end{align*}
}where
{\small  \begin{align*}
    &\boldsymbol\theta(t)\\
    &=\left[0.5+\frac{t}{50},0.75+\frac{t}{75},1.5+0.5\times(-1)^t,\sin\left(\frac{\pi}{4}+\frac{\pi}{100}t\right)\right]^{\top}
\end{align*}}and
{\small  \begin{align*}
&\boldsymbol{f}(x_k(t),u_k(t))\\
&=\left[\frac{x_k(t)\sin(x_k(t))}{1+x_k(t)^2},\mathrm{e}^{\frac{x_k(t)}{100}},u_k(t)^3,\arctan(u_k(t))+u_k(t) \right]^{\top}
\end{align*}}are unknown parameter vector and known function vector, respectively. Note that there are various strong nonlinearities in $\boldsymbol{f}(x_k(t),u_k(t))$, such as polynomial functions and exponential functions, which do not satisfy the global Lipschitz condition. However, as the simulation results will show, our control scheme still achieves high-precision trajectory tracking.

\begin{figure}[!t]
	\centering
	\includegraphics[width=0.45\textwidth]{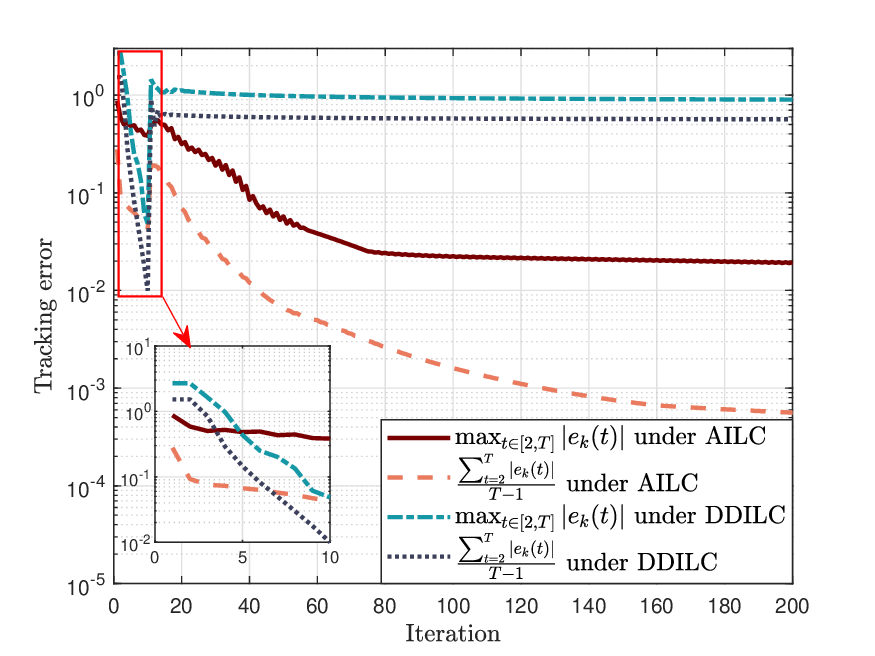}
	\caption{ Maximum and average tracking errors under AILC and DDILC in Example 1.}
	\label{fig:ex1compareerror}
\end{figure}
\begin{figure}[!t]
 \centering
 \subfigure[Tracking profiles for the cosine curve in odd iterations.]{
 \begin{minipage}{0.5\textwidth}
  \centering
  \includegraphics[width=1\textwidth]{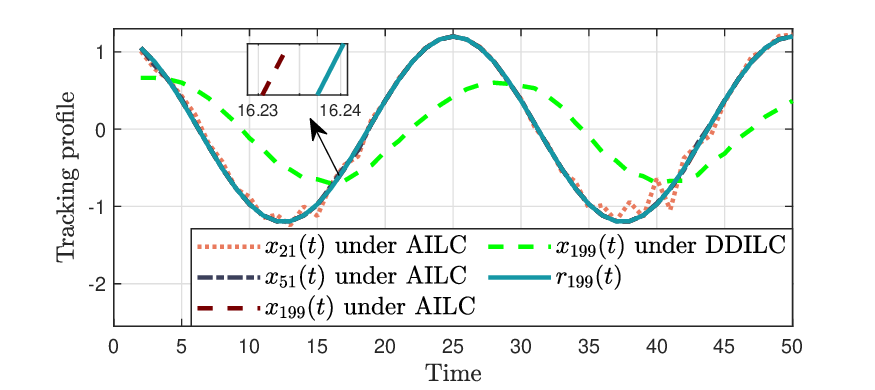}
 \end{minipage}
 }
 \subfigure[Tracking profiles for the sine curve in even iterations.]{

 \begin{minipage}{0.5\textwidth}
 \centering
  \includegraphics[width=1\textwidth]{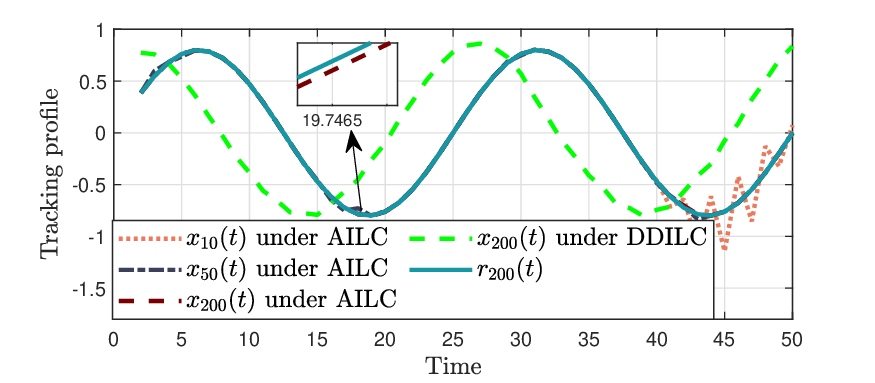}
 \end{minipage}
}
 \caption{Tracking profiles for the iteration-varying reference trajectory under AILC and DDILC in Example 1.} \label{fig:learnability200199}
\end{figure}

First, we compare the proposed AILC with DDILC in the disturbance-free case.  The reference trajectory is set as
\begin{align*}
    r_k(t)=\left\{\begin{array}{ll}
         0.8\sin\left(\frac{2\pi}{25}t \right),& k\leq 10,\ \text{or}\ k\hspace{-2mm}\mod2=0,  \\
        1.2\cos\left(\frac{2\pi}{25}t \right),& k>10,\ \text{and}\ k\hspace{-2mm}\mod2=1.
         \end{array}\right.
\end{align*}
It is evident that $r_k(t)$ is iteration-invariant in the first $10$ iterations, which aims to demonstrate the effectiveness of DDILC and the proposed AILC for invariant reference trajectory. From the $11$th iteration, $r_k(t)$ switches between a cosine curve for odd $k$ and a sine curve for even $k$ in each iteration. The initial state is set as $x_k(1)=0$. Let $\Delta x_k(t)=x_k(t)-x_{k-1}(t)$ and $\Delta u_k(t)=u_k(t)-u_{k-1}(t)$ denote the difference of state or input between adjacent iterations, respectively. The DDILC approach is presented as follows (scheme (2.16)-(2.18) in \cite{chi2022data})
\begin{align*}
u_{k+1}(t)=u_k(t)+\frac{\rho'\hat{\theta}'_{k+1}(t)(r_k(t+1)-x_k(t+1))}
{\lambda'+|\hat{\theta}'_{k+1}(t)|^2},
\end{align*}
where $\hat{\theta}'_{k+1}(t)=\hat{\theta}'_{0}(t)$ if $\sign(\hat{\theta}'_{k+1}(t))\neq \sign(\hat{\theta}'_{0}(t))$ or $|\hat{\theta}'_{k+1}(t)|\leq 10^{-4}$ or $|u_{k}(t)-u_{k-1}(t)|\leq 10^{-4}$, otherwise
\begin{align*}
\hat{\theta}'_{k+1}(t)=\hat{\theta}'_{k}(t)+\frac{\eta'\Delta u_k(t)(\Delta x_k(t+1)-\hat{\theta}'_{k}(t)\Delta u_k(t))}{\mu'+|\Delta u_k(t)|^2}.
\end{align*}
\begin{table}[!t]
        \scriptsize
	\renewcommand{\arraystretch}{0.9}
	\caption{Disturbances in Example 1.}
	\label{table1}
	\centering
	\begin{tabular}{c c c}
		\hline
		 &Disturbance case & Mathematical Expression \\
		\hline\\
        1 &Uniform distribution&  $\mathcal{U}[-0.01,0.01]$\vspace{1mm}\\

        2 & Gaussian distribution&  $\mathcal{N}(0,0.01)$\vspace{1mm}\\

       3 & \makecell{Bernoulli-like\\  distribution}&$\left\{\begin{array}{l}P(w_k(t)=0.03)=0.3\\P(w_k(t)=-0.01)=0.7\end{array}\right.$ \vspace{1mm}\\

       4 & \makecell{Trigonometric\\disturbance}& \makecell{$w_k(t)=0.01\sin(\frac{k\pi t}{50})$\\ \qquad\qquad\qquad$+0.006\cos(\frac{k\pi t}{2})$}\vspace{1mm}\\

     5 & \makecell{High order internal\\model disturbance }& \makecell{$w_k(t)\hspace{-1mm}=\hspace{-1mm}\left\{
	 \begin{array}{l}
\frac{5}{3}w_{k-1}(t)\hspace{-1mm}-\hspace{-1mm}\frac{2}{3}w_{k-2}(t),k>2  \\
-0.02, k=2\\
0.02, k=1
	\end{array}
	\right.$}\vspace{1mm}\\

6 &\makecell{State-dependent\\ disturbance}&
    \makecell{$w_k(t)=0.01x_k(t)$\\ \qquad \quad \qquad\qquad\quad$-0.01\sin(x_k(t)\pi)$}\\
    \hline
	\end{tabular}
\end{table}
\begin{figure}[!t]
 \subfigure[Uniform distribution.]{
       \centering
        \includegraphics[width=0.24\textwidth]{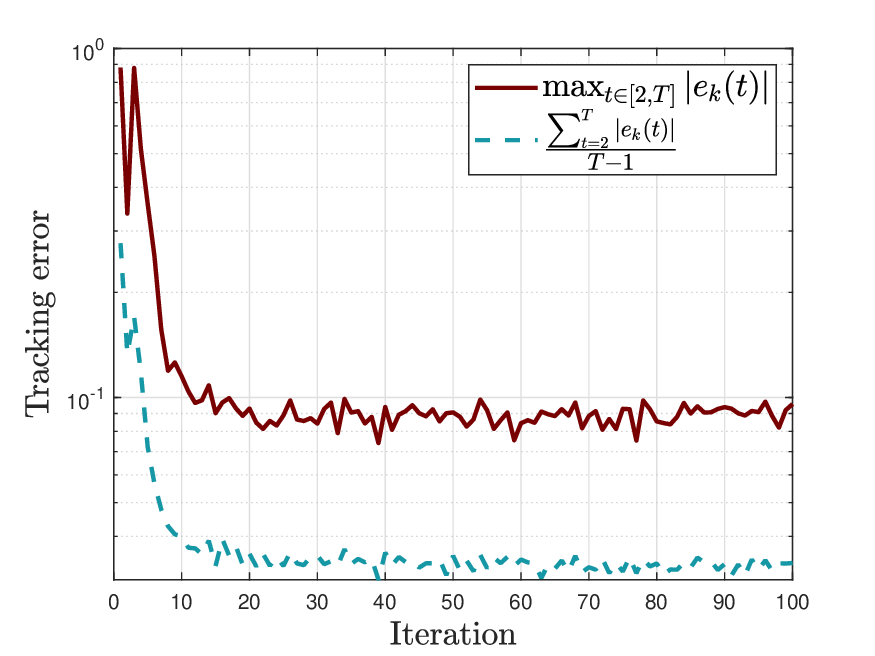}
    }\subfigure[Gaussian distribution.]{
        \centering
        \includegraphics[width=0.24\textwidth]{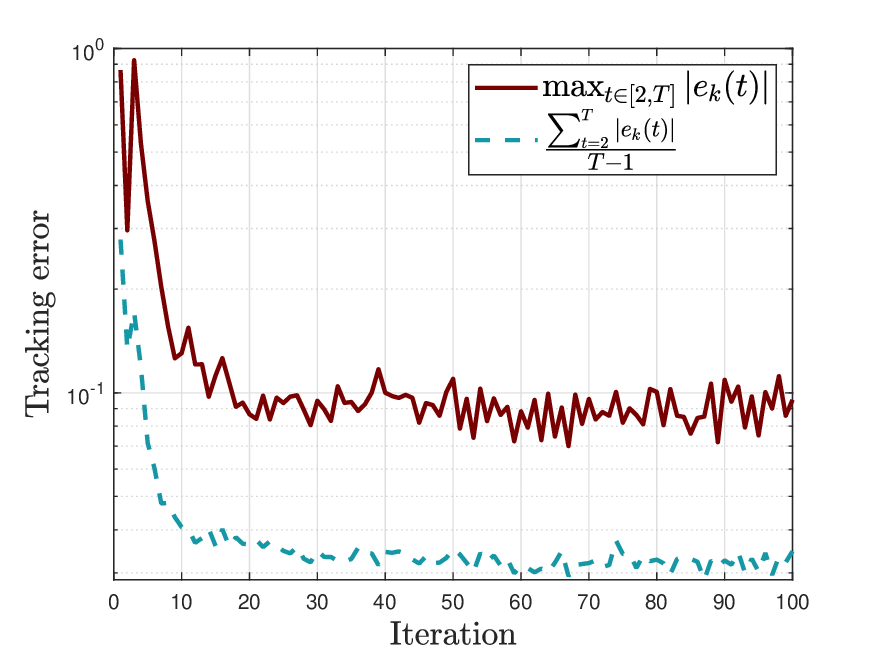}
    }

    \subfigure[Bernoulli-like distribution.]{
        \centering
        \includegraphics[width=0.24\textwidth]{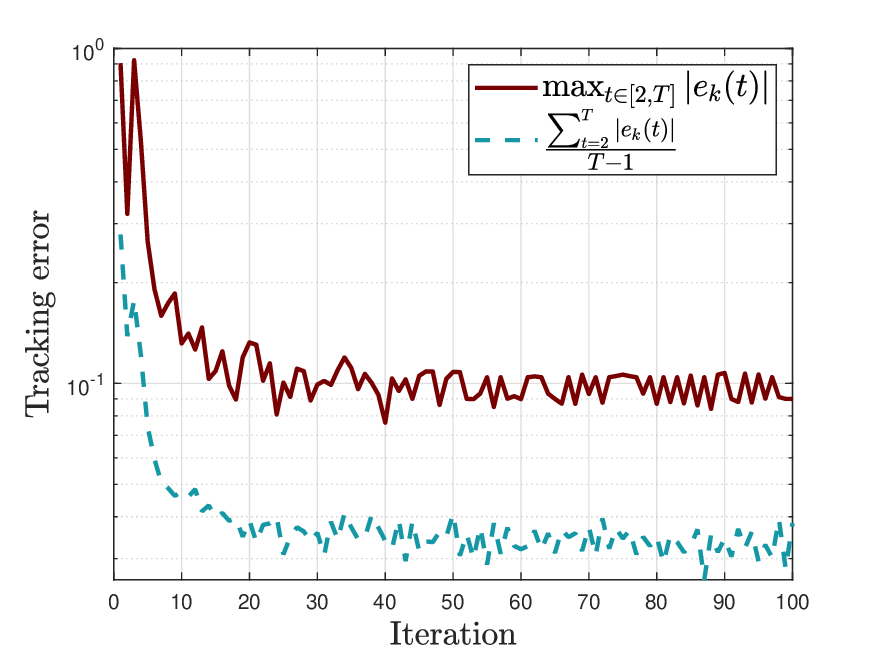}
    }\subfigure[Trigonometric disturbance.]{
       \centering
        \includegraphics[width=0.24\textwidth]{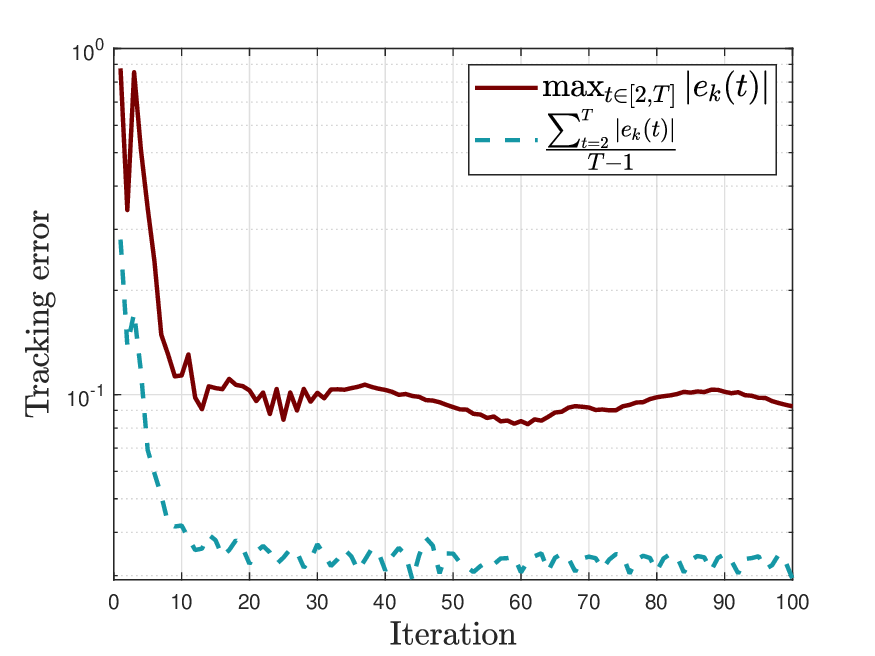}
    }

    \subfigure[High order internal model disturbance.]{
        \centering
        \includegraphics[width=0.24\textwidth]{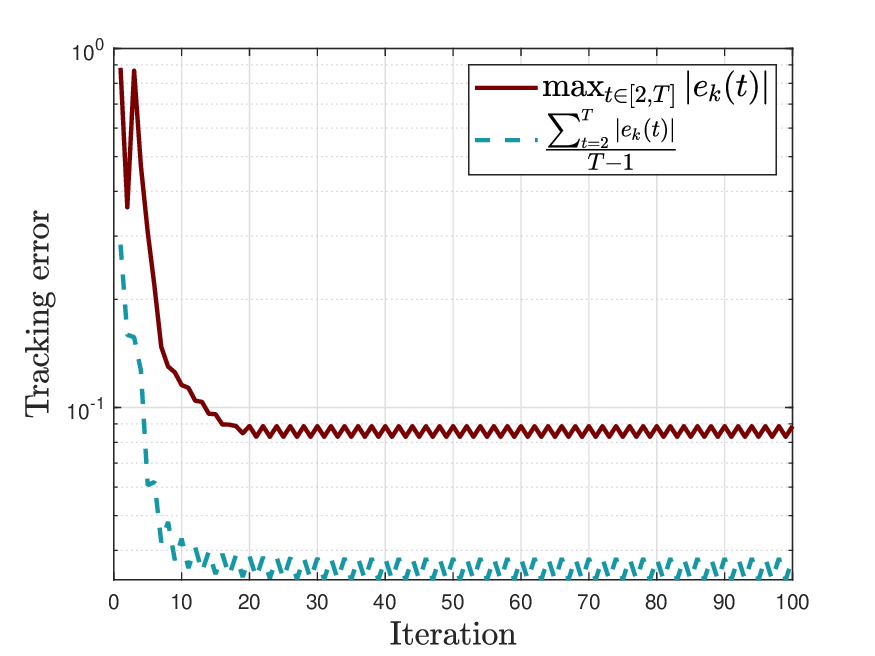}
    }\subfigure[State-dependent disturbance.]{
        \centering
        \includegraphics[width=0.24\textwidth]{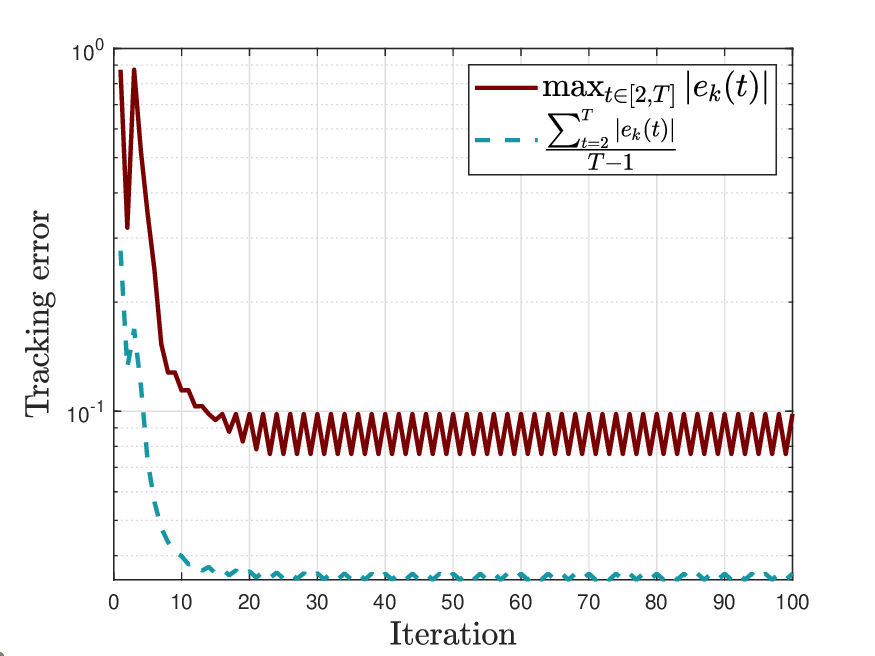}
    }

    \caption{Tracking error under different disturbances in Example 1.}
    \label{fig:ex1errordisturbance}
\end{figure}
The parameters in DDILC approach  are selected as $\eta'=0.5$, $\rho'=0.4$, $\lambda'=1$, $\mu'=0.5$. The initial values of PPD estimation and input are selected as $\hat{\theta}'_{0}(t)=1$, $u_0(t)=0$. For the modified AILC approach \eqref{GDPA1}, \eqref{GDPAmodify}, \eqref{GDPA5}, \eqref{estimateX1}-\eqref{estimateX3}, and \eqref{ALCLaw},
the center and radius of the closed ball used for the projection $\mathcal{B}(\underline{\boldsymbol\theta},R)$ are selected as $\underline{\boldsymbol\theta}=(1,1,1,1)$ and $R=0.9$. The estimation gain is set as $\eta=1.9$, and the initial value of the
parametric estimation is chosen as $\hat{\boldsymbol{\theta}}_0(t)=(1,1,1,1)$. In  Fig.\,\ref{fig:ex1compareerror}, we plot the maximum and average tracking errors on a logarithmic scale along the iteration axis for the DDILC and AILC. It is seen that, during the first $10$ iterations, both DDILC and AILC can reduce the tracking error for the iteration-invariant reference trajectory; however, starting from the $11$th iteration, the contraction mapping process that DDILC relies on must restart at each iteration due to the change of the reference trajectory. Our proposed AILC approach shows significant advantages in tracking the non-repetitive reference trajectory.  Fig.\,\ref{fig:learnability200199} compares the reference and actual trajectories under the DDILC and AILC (note that the relative degree in this example is 1, thus the tracking profiles start from time $t=2$), demonstrating that the proposed AILC approach involves outstanding learning ability and effectively achieves the trajectory tracking goal.

Then, we verify the robustness of the proposed AILC approach \eqref{GDPA1}-\eqref{GDPA5}, \eqref{estimateX1}-\eqref{estimateX3}, and \eqref{ALCLaw} to unknown disturbances. The reference trajectory is set as
\begin{align*}
    r_k(t)=\left\{\begin{array}{ll}
         0.8\sin\left(\frac{2\pi}{25}t \right),& \text{if}~ k\hspace{-2mm}\mod 2=1,  \\
         0.5+0.5\times(-1)^{\lfloor \frac{t}{20}\rfloor},& \text{if}~ k\hspace{-2mm}\mod 2=0.
         \end{array}\right.
\end{align*}It is evident that $r_k(t)$ switches between a sine wave for odd $k$ and a square wave curve for even $k$ in each iteration. The initial state value $x_k(0)$ is selected randomly within $[0,0.01]$. We consider six different types of disturbances as shown in Table \ref{table1} to characterize  various scenarios. The cases 1, 2, and 3 are stochastic disturbances, and cases 4, 5, and 6 are deterministic disturbances. Under the same control parameters, the error convergence profiles corresponding to six different disturbance situations are shown in Fig.\,\ref{fig:ex1errordisturbance}. It is seen that the tracking error converges to the neighborhood of zero. Note that since the characteristic polynomial of the high-order internal mode is stable, the disturbance generated by it converges along the iteration axis, which alleviates the oscillation effect in Fig.\,\ref{fig:ex1errordisturbance} (e). The state-dependent disturbance and stochastic disturbance obeying a Gaussian distribution do not satisfy the boundedness condition (Assumption \ref{assu2}), but as shown in  Fig.\,\ref{fig:ex1errordisturbance} (b) and (f), the AILC scheme is still robust.
\subsection{Example 2}\label{simuex2}

We apply the proposed AILC scheme to a double inverted pendulum system described by the following dynamics \cite{Hovakimyan01012001,802914}:
\begin{align}
\left\{ \begin{array}{l} \label{18}
   \dot{x}_{1,1}=x_{1,2}, \\
   \dot{x}_{1,2}=(\frac{m_1gr}{J_1}-\frac{kr^2}{4J_1})\sin(x_{1,1})+\frac{kr(l-b)}{2J_1}\\
   \qquad\quad+ \frac{kr^2}{4J_1}\sin(x_{2,2})+\frac{u_{1_{\max}}}{J_1}\tanh(u_1),\\
   \dot{x}_{2,1}=x_{2,2},\\
   \dot{x}_{2,2}=(\frac{m_2gr}{J_2}-\frac{kr^2}{4J_2})\sin(x_{2,1})+\frac{kr(l-b)}{2J_2}\\
   \qquad\quad+ \frac{kr^2}{4J_2}\sin(x_{1,2})+\frac{u_{2_{\max}}}{J_2}\tanh(u_2),
   \end{array}\right.
\end{align}
where $x_{1,1}$ and $x_{2,1}$ are the angular displacements of the pendulums from vertical. The physical interpretations of $l$, $b$, $g$, $k$, $r$, $J_i$, $m_i$, and $u_{i_{\max}}$ ($i=1,2$) are given in \cite{Hovakimyan01012001}.
Let $\theta^{i1}=\frac{m_igr}{J_i}-\frac{kr^2}{4J_i}$, $\theta^{i2}=\frac{kr}{2J_i}(l-b)$, $\theta^{i3}=\frac{kr^2}{4J_i}$, $\theta^{i4}=\frac{u_{i_{\max}}}{J_i}$,
$i=1,2$.
The true values of the unknown parameters $\theta^{i1}$, $\theta^{i2}$, $\theta^{i3}$, and $\theta^{i4}$ are calculated based on data in \cite{Hovakimyan01012001} and given as
$\boldsymbol\theta^{(1)}=[\theta^{11},\theta^{12},\theta^{13},\theta^{14}]=[7.12,30,12.5,40]$ and
$\boldsymbol\theta^{(2)}=[\theta^{21},\theta^{22},\theta^{23},\theta^{24}]=[9.62,24,10,32]$.

Using the Euler discretization method and adding the iteration index,\eqref{18} can be rewritten as a discrete-time non-affine nonlinear system with two inputs and two outputs as follows:
\begin{align}\label{ex1rewrite}
\left\{ \begin{array}{l}
x_k^{(1)}(t+2)=2x_k^{(1)}(t+1)-x_k^{(1)}(t)+\theta^{11}\sin(x_k^{(1)}(t))\\
\hspace{22mm}+\theta^{12}+\theta^{13}\sin(x_k^{(2)}(t+1)-x_k^{(2)}(t))\\
\hspace{22mm}+\theta^{14}\tanh(u_k^{(1)}(t)),\\
x_k^{(2)}(t+2)=2x_k^{(2)}(t+1)-x_k^{(2)}(t)+\theta^{21}\sin(x_k^{(2)}(t))\\
\hspace{22mm}+\theta^{22}+\theta^{23}\sin(x_k^{(1)}(t+1)-x_k^{(1)}(t))\\
\hspace{22mm}+\theta^{24}\tanh(u_k^{(2)}(t)),
\end{array}\right.
\end{align}
where $x_k^{(1)}(t)$ and $x_k^{(2)}(t)$ denote the angular displacements of the pendulums with respect to the vertical position, and $t\in [1,50]$. Its relative degree is two. The Jacobian matrix of the MIMO system \eqref{ex1rewrite} with respect to the input is diagonal (and therefore strictly diagonally dominant). It is straightforward to generalize the proposed AILC approach to system \eqref{ex1rewrite}. We use data in the $(k-1)$th iteration to generate parameter estimates $\hat{\boldsymbol\theta}_k^{(i)}(t)$, then, at time instant $t$ of the $k$th iteration, we use signals $\hat{\boldsymbol\theta}_k^{(i)}(t)$, $x_k^{(i)}(t)$, $x_k^{(i)}(t-1)$ to estimate the unavailable state $x_k^{(i)}(t+1)$. The input $u_k(t)$ is obtained by solving the implicit function determined by $x_k^{(i)}(t)$, $x_k^{(i),e}(t+1|t)$, $\hat{\boldsymbol\theta}_k^{(i)}(t)$, and $r_k^{(i)}(t+2)$.

The reference trajectory is
$r_k^{(1)}(t)=r_k^{(2)}(t)=0.1\sin\left(\frac{2\pi t}{25}\right)$, $t\in [1,50]$.
The center and radius of the closed balls
for the projection are set as
$\underline{\boldsymbol\theta}^{(1)}=(7.13, 29.98, 12.52, 39.97)$, $R^{(1)}(t)=0.11$, and
$\underline{\boldsymbol\theta}^{(2)}=(9.63, 24.02, 9.98, 32.02)$, $R^{(2)}(t)=0.11$.
The initial state values $x_k^{(i)}(1)$, $x_k^{(i)}(2)$, $i=1,2,$ are selected randomly within $[0,0.1]$. The initial estimations are set as $\hat{\boldsymbol\theta}_0^{(i)}(t)=0$, $i=1,2$.

\begin{figure}[!t]
	\centering
	\includegraphics[width=0.45\textwidth]{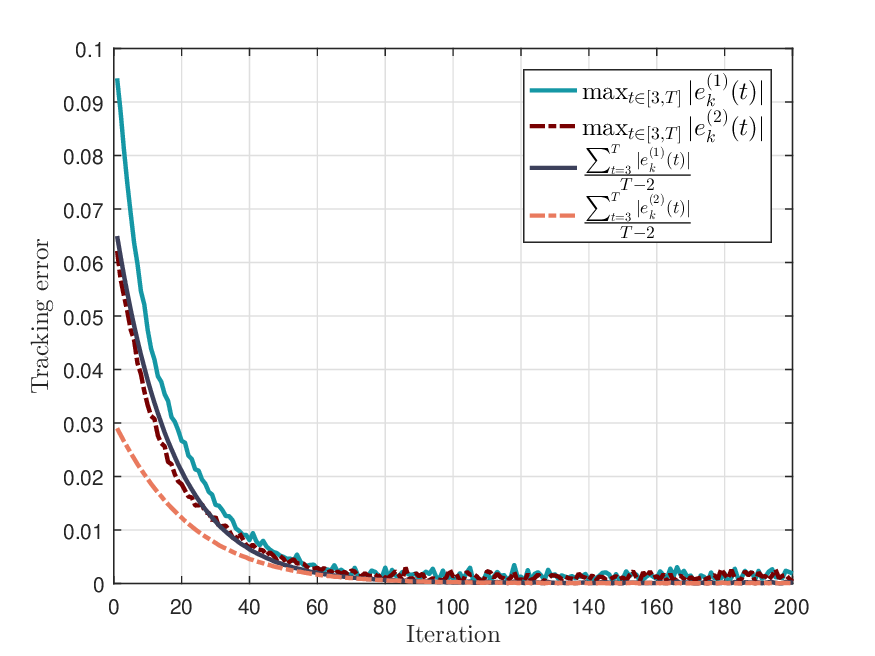}
	\caption{Maximum and average tracking errors of each iteration in Example 2 without disturbance.}
	\label{fig:Highnodiserror}
\end{figure}
\begin{figure}[!t]
 \centering
 \subfigure[Tracking profiles of $x_k^{(1)}(t)$.]{
 \begin{minipage}{0.5\textwidth}
  \centering
  \includegraphics[width=1\textwidth]{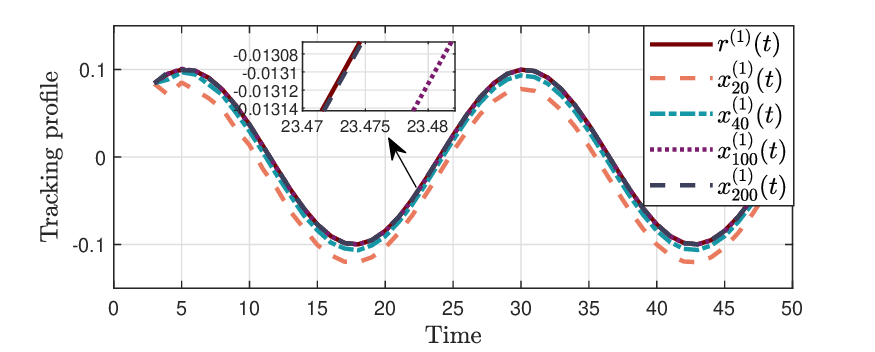}
 \end{minipage}
 }
 \subfigure[Tracking profiles of $x_k^{(2)}(t)$.]{

 \begin{minipage}{0.5\textwidth}
 \centering
  \includegraphics[width=1\textwidth]{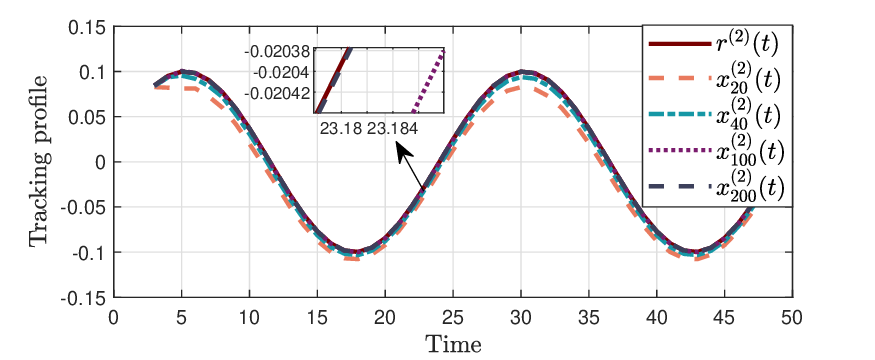}
 \end{minipage}
}
 \caption{Tracking profiles in the 20th, 40th, 100th, and 200th iterations in Example 2 without disturbance.}\label{fig:Highnodistrack}
\end{figure}

We first verify the control performance of the modified AILC scheme \eqref{GDPA1}, \eqref{GDPAmodify}, \eqref{GDPA5}, \eqref{estimateX1}-\eqref{estimateX3}, and \eqref{ALCLaw} in the absence of system disturbances. The control gains are set to $\eta^{(1)}=\eta^{(2)}=0.1$. Fig.\,\ref{fig:Highnodiserror} shows the maximum and average tracking errors as functions of the number of iterations, indicating rapid error reduction. Fig.\,\ref{fig:Highnodistrack} plots the reference and actual trajectories at the 20th, 40th, 100th, and 200th iterations, demonstrating that the actual trajectories converge closely to the reference by the 200th iteration (note that the relative degree in this example is 2, thus the tracking profiles start from time $t=3$). This verifies the learning ability of the proposed AILC scheme.

Next, we use the AILC scheme \eqref{GDPA1}-\eqref{GDPA5}, \eqref{estimateX1}-\eqref{estimateX3}, and \eqref{ALCLaw} to verify control performance in the presence of system disturbances. The system disturbances are modeled as
$w_k^{1}(t)=1\times10^{-4}\cos(kt\pi)+1\times10^{-4}\sin\left(\frac{kt\pi}{2}\right)$ and
$w_k^{2}(t)=1\times10^{-4}\cos(2kt\pi)+1\times10^{-4}\sin(kt\pi)$.
The selection of the control parameters remains unchanged.
Fig.\,\ref{fig:Highdiserror} shows the convergence of the maximum and average tracking errors over iterations, and Fig.\,\ref{fig:Highdistrack200} illustrates the reference and actual trajectories at the 100th iteration, revealing that while tracking is generally successful, accuracy is somewhat compromised by the system disturbances.

\begin{figure}[!t]
	\centering
	\includegraphics[width=0.45\textwidth]{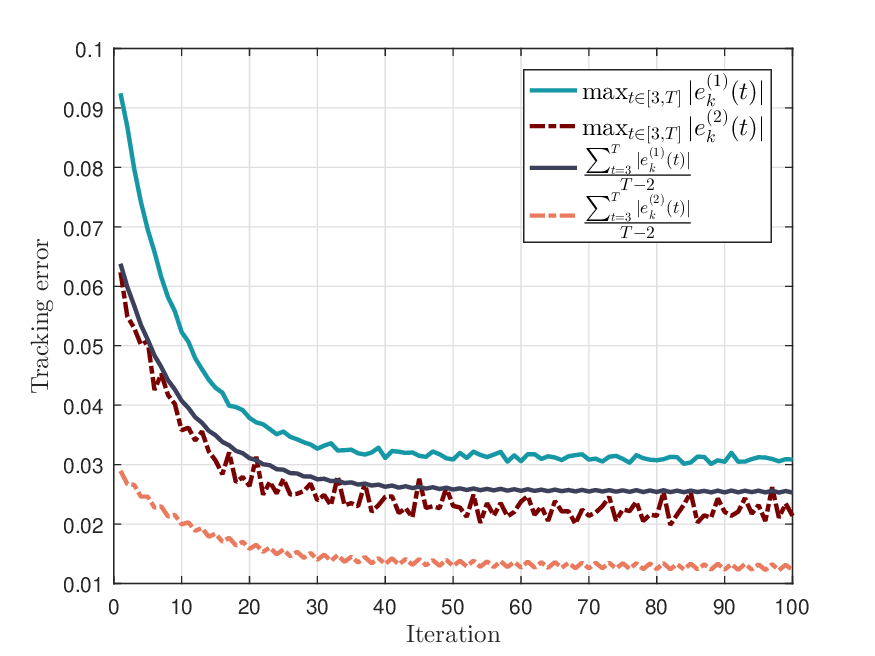}
	\caption{ Maximum and average tracking errors of each iteration in Example 2 in the presence of disturbances.}
	\label{fig:Highdiserror}
\end{figure}

\section{Conclusions}\label{conclu}

This paper deals with the trajectory tracking problem for a class of non-affine nonlinear systems. A rigorously designed AILC scheme, which integrates a GDPA law and a state estimator is proposed, departing significantly from the existing ILC methodologies that rely on dynamic linearization techniques or neural networks. It comprehensively addresses various practical considerations, including high relative degrees, system uncertainties and disturbances, iteration-varying reference trajectories. We also proposed a numerical approach that approximates the control input, up to a desired accuracy, in the case where an explicit AILC input for the non-affine system cannot be obtained. A rigorous analysis of the tracking performance, substantiated by numerical simulations, is also provided.

Our future work will focus on relaxing the global Lipschitz condition. In particular, two possible technical routes may be pursued: the direct one is to develop new convergence analysis tools that can handle nonlinear relationships between variables without relying on the key technical lemma, and the indirect one is to transform the local Lipschitz system into a global one in a compact set by ensuring the boundedness of system signals first, for example via input saturation or barrier functions.

\begin{figure}[!t]
	\centering
	\includegraphics[width=0.45\textwidth]{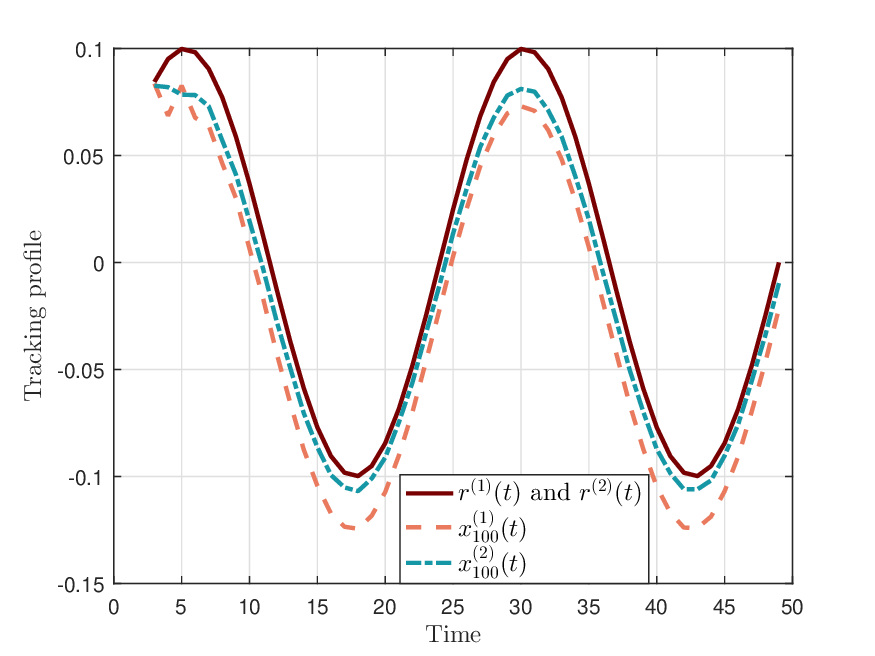}
	\caption{Tracking profiles of the 100th iteration in Example 2 in the presence of disturbances.}
	\label{fig:Highdistrack200}
\end{figure}

\section*{Appendix}
{\subsection{Proof of Proposition \ref{pp1}}\label{APPendixA}

Without loss of any generality, we consider the case of $\boldsymbol\theta(t)^\top\frac{\partial }{\partial u}\boldsymbol{f}(\boldsymbol{X}_k(t),u)> d_0$. Due to the continuity of $\boldsymbol\theta(t)^\top\frac{\partial }{\partial u}\boldsymbol{f}(\boldsymbol{X}_k(t),u)$, there exists a constant $M>0$ such that $\sup_{u\in \mathcal{B}(0,\frac{|c|}{d_0})}|\boldsymbol\theta(t)^\top\frac{\partial }{\partial u}\boldsymbol{f}(\boldsymbol{X}_k(t),u)|<M$. Choosing $l>M$, one obtains
\begin{align*}
    \left|1-\frac{1}{l}\boldsymbol\theta(t)^\top \frac{\partial }{\partial u}\boldsymbol{f}(\boldsymbol{X}_k(t),u)\right| \leq 1-\frac{d_0}{l},\ \forall\ u\in \mathcal{B}\left(0,\frac{|c|}{d_0}\right).
\end{align*}

Therefore, for any $u \in \mathcal{B}\Big(0,\frac{|c|}{d_0}\Big)$, by virtue of the differential mean value theorem, one obtains
\begin{align*}
|\mathcal{T}(u)|&=\Big|u-\frac{\mathcal{Z}(0)}{l}-\frac{1}{l}(\mathcal{Z}(u)-\mathcal{Z}(0))\Big| \nonumber\\
&=\left|u-\frac{c}{l}-\boldsymbol\theta(t)^\top\frac{\partial }{\partial u}\boldsymbol{f}(\boldsymbol{X}_k(t),u)\big|_{u=\xi} \frac{u}{l} \right|\\
&=\frac{|c|}{l}+\left|(1-\frac{1}{l} \boldsymbol\theta(t)^\top\frac{\partial }{\partial u}\boldsymbol{f}(\boldsymbol{X}_k(t),u)\big|_{u=\xi})u   \right| \\
&\leq \frac{|c|}{l}+(1-\frac{d_{0}}{l})\frac{|c|}{d_0}\\
&=\frac{|c|}{d_0},
\end{align*}
where $\xi$ is a mean value satisfying $|\xi|\leq |u|$. Therefore, the mapping $\mathcal{T}$ maps $\mathcal{B}\Big(0,\frac{|c|}{d_0}\Big)$ to itself. For any $u_1,\ u_2\in\mathcal{B}\Big(0,\frac{|c|}{d_0}\Big)$, the differential mean value theorem yields
\begin{align}\label{CM}
    &\Big|\mathcal{T}(u_1)-\mathcal{T}(u_2) \Big|\nonumber\\
    &\leq   \Big| 1-\frac{1}{l}\boldsymbol\theta(t)^\top\frac{\partial }{\partial u}\boldsymbol{f}(\boldsymbol{X}_k(t),u)\big|_{u=\xi'}\Big|~|u_1-u_2|\nonumber\\
    &\leq (1-\frac{d_0}{l})|u_1-u_2|,
\end{align}
where $\xi'$ is a mean value satisfying $|\xi'|\leq \max\{|u_1|,|u_2|\}$. From \eqref{CM}, $\mathcal{T}$ is a contraction mapping. According to Banach's fixed point theorem, there exists a unique $u_k^\ast(t)\in \mathcal{B}\Big(0,\frac{|c|}{d_0}\Big)$ such that $ \mathcal{T}(u_k^\ast(t))=u_k^\ast(t)$, which means $u_k^\ast(t)$ satisfies \eqref{trackingequation}. The proof is complete. \hfill$\blacksquare$

\subsection{Proof of Theorem \ref{Th1}} \label{APPendixB}

In this proof, we use a capital $C$ to represent constants independent of $k$.
\begin{lemma} \label{lemmaestimate}
Under the conditions of Theorem \ref{Th1}, for each $t\in\{0,\cdots,T-\rho\}$, the following statements hold:

(1) The estimates $\hat{\boldsymbol\theta}_k(t)$ and $\hat{w}_k(t)$ are bounded.

(2) As $k$ approaches infinity, $a_k (t)^2\epsilon_k (t+\rho)^2m_k (t)^2$ and $\hat{w}_k(t)$ converge to zero and a specific positive constant, respectively.
\end{lemma}

\emph{Proof of Lemma \ref{lemmaestimate}:~}  Let $\Tilde{\boldsymbol\theta}_k (t)=\hat{\boldsymbol\theta}_k (t) -\boldsymbol\theta (t)$ and $\Tilde{w}_{k} (t)=\hat{w} _{k}(t)-w$. One has
\begin{align*}
  \Tilde{\boldsymbol\theta}_{k+1} (t)&=\Tilde{\boldsymbol\theta}_{k} (t)+\eta a_k(t)\epsilon_k(t+\rho)\boldsymbol{f}(\boldsymbol{X}_k(t),u_k(t))+\boldsymbol{g}_k (t),
\end{align*}
where
\begin{align*}
    \boldsymbol{g}_k (t)=\left\{ \begin{array}{l}
   \boldsymbol{0},\  \text{if}\  \Bar{\boldsymbol\theta}_{k+1} (t)\in \mathcal{B}_2(\underline{\boldsymbol\theta}, R), \\
	\left(\frac{R}{\|\Bar{\boldsymbol\theta}_{k+1} (t)-\underline{\boldsymbol\theta}\|}-1\right)(\Bar{\boldsymbol\theta}_{k+1} (t)-\underline{\boldsymbol\theta}),\ \text{otherwise.} \end{array}\right.
\end{align*}

Note that
\begin{align}\label{differenceW}
    \Tilde{w}_{k+1}(t)=\Tilde{w}_{k} (t)+\eta a_k(t)|\epsilon_k (t+\rho)|.
\end{align}

Define a Lyapunov-like positive-definite function
\begin{align*}
    V_k (t)=\frac{1}{2}\Tilde{\boldsymbol\theta}_k (t)^{\top}\Tilde{\boldsymbol\theta}_k (t)+\frac{1}{2}\Tilde{w}_{k} (t)^2.
\end{align*}

The difference between adjacent iterations of the first part in $V_k (t)$ is given by
\begin{align}\label{differencetheta}
    &\frac{1}{2}\Tilde{\boldsymbol\theta}_{k+1} (t)^{\top}\Tilde{\boldsymbol\theta}_{k+1} (t)-\frac{1}{2}\Tilde{\boldsymbol\theta}_k (t)^{\top}\Tilde{\boldsymbol\theta}_k (t)\nonumber\\
    &=\frac{1}{2}a_k (t)^2\eta ^2\epsilon_k (t+\rho)^2\|\boldsymbol{f}(\boldsymbol{X}_k(t),u_k(t))\|^2\nonumber \\
    &\quad+\Tilde{\boldsymbol\theta}_{k} (t)^{\top}\boldsymbol{f}(\boldsymbol{X}_k(t),u_k(t))a_k (t)\eta \epsilon_k (t+\rho)\nonumber\\
    &\quad+\frac{1}{2}\|\boldsymbol{g}_k (t)\|^2+\boldsymbol{g}_k (t)^{\top}(\Bar{\boldsymbol\theta}_{k+1} (t)-\boldsymbol\theta (t)).
\end{align}

If $\Bar{\boldsymbol\theta}_{k+1}(t)\notin \mathcal{B}_2(\underline{\boldsymbol\theta},R)$, one obtains
\begin{align} \label{g+gtheta}
&\frac{1}{2}\|\boldsymbol{g}_k (t)\|^2+\boldsymbol{g}_k (t)^{\top}(\Bar{\boldsymbol\theta}_{k+1} (t)-\boldsymbol\theta (t))\nonumber \\
&=\frac{1}{2} \left(1-\frac{R}{\|\Bar{\boldsymbol\theta}_{k+1} (t)-\underline{\boldsymbol\theta}\|}\right)^2\|\Bar{\boldsymbol\theta}_{k+1} (t)-\underline{\boldsymbol\theta}\|^2\nonumber
\\&\quad-\left(1-\frac{R}{\|\Bar{\boldsymbol\theta}_{k+1} (t)-\underline{\boldsymbol\theta}\|}\right) \left(\Bar{\boldsymbol\theta}_{k+1} (t)-\underline{\boldsymbol\theta}\right)^{\top}\nonumber\\
&\quad \times\left(\Bar{\boldsymbol\theta}_{k+1} (t)-\underline{\boldsymbol\theta}+\underline{\boldsymbol\theta}-\boldsymbol\theta (t)\right)\nonumber\\
&\leq \frac{1}{2}\left(\|\Bar{\boldsymbol\theta}_{k+1} (t)-\underline{\boldsymbol\theta}\|-R\right)^2\nonumber\\
&\quad-\Big(\|\Bar{\boldsymbol\theta}_{k+1} (t)-\underline{\boldsymbol\theta}\|-R\Big)\Big(\|\Bar{\boldsymbol\theta}_{k+1} (t)-\underline{\boldsymbol\theta}\|\nonumber\\
&\quad-\|\boldsymbol\theta (t)-\underline{\boldsymbol\theta}\|\Big)\nonumber\\
&\leq \Big(\|\Bar{\boldsymbol\theta}_{k+1} (t)-\underline{\boldsymbol\theta}\|-R\Big) \Big(\|\boldsymbol\theta (t)-\underline{\boldsymbol\theta}\|-R\Big)\nonumber\\
&\leq 0.
\end{align}

If $\Bar{\boldsymbol\theta}_{k+1} (t)\in \mathcal{B}_2(\underline{\boldsymbol\theta},R)$, $\boldsymbol{g}_k (t)=\boldsymbol{0}$ and \eqref{g+gtheta} still holds. Therefore, from \eqref{differencetheta} and \eqref{g+gtheta}, one obtains
\begin{align*}
         &\frac{1}{2}\Tilde{\boldsymbol\theta}_{k+1} (t)^{\top}\Tilde{\boldsymbol\theta}_{k+1} (t)-\frac{1}{2}\Tilde{\boldsymbol\theta}_k (t)^{\top}\Tilde{\boldsymbol\theta}_k (t)\nonumber\\
    &\leq \frac{1}{2}a_k (t)^2\eta ^2\epsilon_k (t+\rho)^2\|\boldsymbol{f}(\boldsymbol{X}_k(t),u_k(t))\|^2\nonumber \\
    &\quad+\Tilde{\boldsymbol\theta}_{k} (t)^{\top}\boldsymbol{f}(\boldsymbol{X}_k(t),u_k(t))a_k (t)\eta \epsilon_k (t+\rho).
\end{align*}

From \eqref{system} and \eqref{GDPA1}, one gets
\begin{align*}
        \Tilde{\boldsymbol\theta}_{k} (t)^{\top}\boldsymbol{f}(\boldsymbol{X}_k(t),u_k(t))=-m_k (t)^2\epsilon_k(t+\rho)+w_k(t).
\end{align*}
Therefore,
\begin{align*}
       &\frac{1}{2}\Tilde{\boldsymbol\theta}_{k+1} (t)^{\top}\Tilde{\boldsymbol\theta}_{k+1} (t)-\frac{1}{2}\Tilde{\boldsymbol\theta}_k (t)^{\top}\Tilde{\boldsymbol\theta}_k (t)\nonumber\\
   &\leq\frac{1}{2}a_k (t)^2\eta ^2\epsilon_k (t+\rho)^2\|\boldsymbol{f}(\boldsymbol{X}_k(t),u_k(t))\|^2-a_k (t)\nonumber\\
   &\quad\times\eta m_k (t)^2\left(\epsilon_k (t+\rho)^2-\frac{|\epsilon_k(t+\rho)|w }{m_k(t)^2}\right)\nonumber \\
   &=\frac{1}{2}a_k (t)^2\eta ^2\epsilon_k (t+\rho)^2\|\boldsymbol{f}(\boldsymbol{X}_k(t),u_k(t))\|^2-a_k (t)\nonumber\\
   &\quad\times \eta  m_k (t)^2\left(\epsilon_k (t+\rho)^2-\frac{|\epsilon_k (t+\rho)|\hat{w}_k (t)}{m_k (t)^2}\right)\nonumber\\
   &\quad-a_k (t)\eta |\epsilon_k (t+\rho)|\Tilde{w}_k (t).
\end{align*}

If $|\epsilon_k (t+\rho)|>\frac{\hat{w}_k (t)}{m_k (t)^2}$, one obtains
\begin{align}\label{aepsilonm}
       &a_k (t)^2\epsilon_k (t+\rho)^2m_k (t)^2 \nonumber\\
   &=a_k (t)\epsilon_k (t+\rho)^2m_k (t)^2\left(1-\frac{\hat{w}_k (t)}{|\epsilon_k (t+\rho)|m_k (t)^2}\right) \nonumber\\
   &=a_k (t)m_k (t)^2\left(\epsilon_k(t+\rho)^2-\frac{\hat{w}_k (t)|\epsilon_k (t+\rho)|}{m_k (t)^2}\right).
\end{align}

If $|\epsilon_k (t+\rho)|\leq \frac{\hat{w}_{k} (t)}{m_k (t)^2}$, $a_k (t)=0$ and \eqref{aepsilonm} still holds. Therefore, one has
\begin{align}\label{differencetheta2}
    &\frac{1}{2}\Tilde{\boldsymbol\theta}_{k+1} (t)^{\top}\Tilde{\boldsymbol\theta}_{k+1} (t)-\frac{1}{2}\Tilde{\boldsymbol\theta}_k (t)^{\top}\Tilde{\boldsymbol\theta}_k (t)\nonumber\\
&\leq \frac{1}{2}a_k (t)^2\eta ^2\epsilon_k (t+\rho)^2\|\boldsymbol{f}(\boldsymbol{X}_k(t),u_k(t))\|^2\nonumber\\
&\quad-a_k(t)\eta |\epsilon_k(t+\rho)|\Tilde{w}_k (t)\nonumber\\
&\quad-\eta a_k (t)^2\epsilon_k (t+\rho)^2m_k(t)^2.
\end{align}

For the second part of $V_k(t)$, from \eqref{differenceW}, one has
\begin{align}\label{differencew2}
    \frac{1}{2}\Tilde{w}_{k+1}(t)^2-\frac{1}{2}\Tilde{w}_{k} (t)^2=&\frac{1}{2}a_k (t)^2\eta ^2\epsilon_k (t+\rho)^2\nonumber\\
&+\Tilde{w}_{k} (t)\eta a_k (t)|\epsilon_k(t+\rho)|.
\end{align}

Furthermore, \eqref{differencetheta2} and \eqref{differencew2} yield
\begin{align}\label{decrease}
    \Delta V_{k+1} (t)&=V_{k+1}(t)- V_{k}(t)\nonumber\\
&\leq \left(\frac{1}{2}\eta ^2-\eta \right)a_k (t)^2\epsilon_k (t+\rho)^2m_k (t)^2.
\end{align}

If $0<\eta <2$, then $\frac{1}{2}\eta ^2-\eta<0$, and $V_k(t)$ is monotonically decreasing along the iteration axis. Therefore, $\Tilde{\boldsymbol\theta}_{k}(t)$ and $\hat{w}_k(t)$ are bounded. The first claim of Lemma \ref{lemmaestimate} is proved.

From \eqref{decrease}, one obtains
\begin{align*}
(\eta-\frac{1}{2}\eta^2)a_k (t)^2\epsilon_k (t+\rho)^2m_k (t)^2\leq V_k(t)-V_{k+1}(t).
\end{align*}
Summing both sides from $k=1$ to $k'$ yields
\begin{align}\label{sum}
\sum_{k=1}^{k'}(\eta-\frac{1}{2}\eta^2)a_k (t)^2\epsilon_k (t+\rho)^2m_k (t)^2\leq V_1(t)-V_{k'+1}(t).
\end{align}
Inequality \eqref{sum} indicates that the series $\sum_{k=1}^{\infty}a_k (t)^2\epsilon_k (t+\rho)^2m_k (t)^2$ converges. Thus,
\begin{align*}
       \lim_{k\to \infty}a_k (t)^2\epsilon_k (t+\rho)^2m_k (t)^2=0.
\end{align*}
In addition, according to \eqref{GDPA4}, the sequence $\{\hat{w}_k(t)\}_{k=1}^{\infty}$ is monotonically increasing and bounded. Therefore $\lim_{k\to\infty} \hat{w}_k(t)$ exists.
Lemma \ref{lemmaestimate} is proven. \hfill$\blacksquare$

The following definition characterizes the \emph{linear boundedness} for further analysis.
\begin{definition}
For variable sequences $\{a_k\}$, $\{b_k^{(1)}\}$, $\{b_k^{(2)}\} $,$\cdots$,$\{b_k^{(n)}\}$, we say that $\{a_k\}$ is \emph{linearly bounded} by $\{b_k^{(1)}\}$, $\{b_k^{(2)}\} $,$\cdots$,$\{b_k^{(n)}\}$, if there exist positive constants $C^{(1)}$, $C^{(2)}$,$\cdots$,$C^{(n)}$ that are independent of $k$, such that
\begin{align*}
\|a_k\| \leq  C^{(1)}\|b_k^{(1)}\|+ C^{(2)}\|b_k^{(2)}\|+\cdots+C^{(n)}\|b_k^{(n)}\|.
\end{align*}
\end{definition}
\begin{lemma}\label{lemmabound}
Under the conditions of Theorem \ref{Th1}, for each $t\in\{0,\cdots,T-\rho\}$, the following statements hold:

(1) Variable $m_k(0)$ is bounded. For $t\geq 1$, $m_k(t)$ is linearly bounded by $e_k(t+\rho-1), e_k(t+\rho-2), \cdots, e_k(\rho)$.

(2) There exist constants $K^{0,t},K^{1,t},\cdots,K^{t-1,t}$ such that
\begin{align*}
|e_k(t+\rho)|&<\varpi_k(t+\rho,\cdots,t)+\sum_{i=0}^{t-1}K^{i,t} |\varpi_k(i+\rho,\cdots,i)|,
\end{align*}
where $\varpi_k(t+\rho,\cdots,t)=x_k(t+\rho)-\hat{\boldsymbol\theta}_k(t)^{\top}\boldsymbol{f}(\boldsymbol{X}_k(t),u_k(t))$.
\end{lemma}

\emph{Proof of Lemma \ref{lemmabound}}: From the definition of $u_k(t)$ in \eqref{ALCLaw}, one has
\begin{align}\label{ufbounded}
    |u_k(t)| &\leq \frac{|c'|}{d_0}
    \leq C(\|\hat{\boldsymbol\theta}_k(t)\| \|\boldsymbol{f}(\boldsymbol{X}^e_k(t),0)\|+|r_k(t+\rho)|)\nonumber\\
    &\leq C \|\boldsymbol{f}(\boldsymbol{X}^e_k(t),0)\|+C.
\end{align}
The last inequality holds as per the fact that $\|\hat{\boldsymbol\theta}_k(t)\|$ and $|r_k(t+\rho)|$ are bounded. Combining Assumption \ref{assu3} with \eqref{ufbounded}, it is clear that $u_k(t)$ is linearly bounded by $\boldsymbol{X}_k^e(t)$.

\emph{Case of $t<\rho$}: According to Assumption \ref{assu4} and the state estimation setting, $x_k^e(t|t')=x_k(t)$ holds for any $0\leq t\leq \rho-1$ and $0\leq t'\leq t$. Since the vector $\boldsymbol{X}_k^e(0)=[x_k^e(\rho-1|0)$, $\cdots$, $x_k^e(1|0)$, $x_k(0)]^\top=[x_k(\rho-1),\cdots,x_k(1),x_k(0)]^\top$ is bounded by Assumption \ref{assu4}, it follows that $u_k(0)$ is bounded. The vector $\boldsymbol{X}_k^e(1)=[x_k^e(\rho|1),x_k(\rho-1),\cdots,x_k(1)]^\top$ is bounded as long as $x_k^e(\rho|1)$ is bounded. From \eqref{estimateX1}, $x_k^e(\rho|1)$ is bounded due to the boundedness of $x_k^e(\rho-1|1)=x_k(\rho-1),\cdots,x_k^e(2|1)=x_k(2),x_k(1),x_k(0)$ and $u_k(0)$. The boundedness of $u_k(1)$ follows from the boundedness of $\boldsymbol{X}_k^e(1)$. By recursion, the boundedness of $u_k(2),\cdots,u_k(\rho-1)$ can be established.

\emph{Case of $t\geq \rho$}: As per \eqref{estimateX2}, $x_k^e(t+1|t)$ is linearly bounded by $x_k(t),\cdots,x_k(t-\rho+1)$ and $u_k(t-\rho+1)$. Moreover, $x_k^e(t+2|t)$ is linearly bounded by $x_k^e(t+1|t),x_k(t),\cdots,x_k(t-\rho+2)$ and $u_k(t-\rho+2)$. By recursion, it follows that $x^e_k(t+\rho-1|t),x^e_k(t+\rho-2|t),\cdots,x^e_k(t+1|t)$ are all linearly bounded by $x_k(t),\cdots,x_k(t-\rho+1)$ and  $u_k(t-1),u_k(t-2),\cdots,u_k(t-\rho+1)$. Therefore,
\begin{align*}
  |u_k(t)|&  \leq  C+  \sum_{i=1}^\rho C x_k(t-\rho+i) + \sum_{i=1}^{\rho-1} C u_k(t-\rho+i).
\end{align*}
Similarly,
\begin{align*}
  |u_k(t-1)|&
  \leq C+  \sum_{i=1}^\rho C |x_k(t-\rho-1+i)|\\
  &\quad + \sum_{i=1}^{\rho-1} C|u_k(t-\rho-1+i)|.
\end{align*}
By recursion,
\begin{align}\label{ukacbound}
  |u_k(t)|&  \leq  C+\sum_{i=1}^t C|x_k(i)| + \sum_{i=1}^{\rho-1} C|u_k(i)|.
\end{align}
From \eqref{ukacbound} and the boundedness of $u_{k}(0),\cdots,u_{k}(\rho-1)$, $x_{k}(0),\cdots,x_{k}(\rho-1)$, one concludes that $u_k(t)$ is linearly bounded by $x_k(t),x_k(t-1),\cdots,x_k(\rho)$.

From the definition of $m_k(t)$ and the Lipschitz continuity of $\boldsymbol{f}(\boldsymbol{X},u)$, it follows that $m_k(t)$ is linearly bounded by $x_k(t+\rho-1),\cdots,x_k(t)$ and $u_k(t)$. Based on the previous discussion on the upper bound of $u_k(t)$, it is clear that $m_k(0)$ is bounded. For $t\geq 1$, one has
\begin{align*}
|m_k (t)| \leq C+\sum_{i=\rho}^{t+\rho-1}C|x_k(i)| \leq C+\sum_{i=\rho}^{t+\rho-1}C|e_k(i)|,
\end{align*}
where the last inequality holds due to the fact that $x_k(\cdot)=r_k(\cdot)+e_k(\cdot)$ and the reference trajectory $r_k(\cdot)$ is bounded. The first claim of Lemma \ref{lemmabound} is proved.

For the second claim, one has
\begin{align}\label{e(t+rho)zhon}
|e_k(t+\rho)|&=|x_k(t+\rho)-r_k(t+\rho)|\nonumber\\
&\leq |x_k(t+\rho)-\hat{\boldsymbol\theta}_k(t)^{\top}\boldsymbol{f}(\boldsymbol{X}_k(t),u_k(t))|
\nonumber \\
&\quad+|\hat{\boldsymbol\theta}_k(t)^{\top}\boldsymbol{f}(\boldsymbol{X}_k(t),u_k(t))-r_k(t+\rho)|\nonumber\\
&\leq|\varpi_k(t+\rho,\cdots,t)|+C\|\boldsymbol{X}_k(t)-\boldsymbol X_k^{e}(t)\|\nonumber\\
&\leq\hspace{-1mm} |\varpi_k(t\hspace{-0.2mm}\hspace{-0.2mm}+\hspace{-0.2mm}\rho,\cdots,t)|\hspace{-0.5mm}+\hspace{-0.5mm}\sum_{i=1}^{\rho-1}C|x_k(t\hspace{-0.2mm}+\hspace{-0.2mm}i)\hspace{-1mm}-\hspace{-1mm}x_k^e(t\hspace{-0.2mm}+\hspace{-0.2mm}i|t)|  \nonumber\\
&\leq |\varpi_k(t+\rho,\cdots,t)|+\sum_{i=1}^{\rho-1}C |e_k(t+i)|\nonumber\\
&\quad+\sum_{i=1}^{\rho-1} C|r_k(t+i)-x_k^e(t+i|t)|.
\end{align}

Let us look at the term $|r_k(t+\rho-1)-x_k^e(t+\rho-1|t)|$. From \eqref{estimateX1} and \eqref{pred-model}, one obtains
\begin{align}\label{RHS}
&|r_k(t+\rho-1)-x_k^e(t+\rho-1|t)|\nonumber\\
&\leq \sum_{i=1}^{\rho-2} C|x_k^e(t+i|t)-x_k^e(t+i|t-1)|\nonumber\\
&\quad +C|x_k(t)-x_k^e(t|t-1)|.
\end{align}

Due to the Lipschitz condition, one has
\begin{align*}
&|x_k^e(t+1|t)-x_k^e(t+1|t-1)|\leq C|x_k(t)-x_k^e(t|t-1)|,\\
&|x_k^e(t+2|t)-x_k^e(t+2|t-1)| \\
&\leq C|x_k^e(t+1|t)-x_k^e(t+1|t-1)| +C|x_k(t)-x_k^e(t|t-1)|\\
&\leq C|x_k(t)-x_k^e(t|t-1)|.
\end{align*}
By recursion, each term on the right-hand side of \eqref{RHS} is less than $C|x_k(t)-x_k^e(t|t-1)|$, where $C$ is a constant independent of $k$. Therefore, it follows that $|r_k(t+\rho-1)-x_k^e(t+\rho-1|t)|\leq C|x_k(t)-x_k^e(t|t-1)|$.

Similarly, one can verify that, for $i=1,\cdots,\rho-1$,
\begin{align}\label{r-x}
&|r_k(t+\rho-i)-x_k^e(t+\rho-i|t)|\nonumber \\
&\quad\leq\sum_{j=1}^iC|x_k(t-i+j)-x_k^e(t-i+j|t-i)|.
\end{align}

Substituting \eqref{r-x} to \eqref{e(t+rho)zhon} yields
\begin{align}\label{ekt+rhozhon2}
& |e_k(t+\rho)|\leq |\varpi_k(t+\rho,\cdots,t)|+\sum_{i=1}^{\rho-1} C|e_k(t+i)| \nonumber\\
&\qquad+\sum_{i=1}^{\rho-1}\sum_{j=1}^iC|x_k(t-i+j)-x_k^e(t-i+j|t-i)|\nonumber\\
&\quad\leq |\varpi_k(t+\rho,\cdots,t)|+\sum_{i=1}^{2\rho-2} C|e_k(t-\rho+1+i)|\nonumber \\
&\qquad+\sum_{i=1}^{\rho-1}\sum_{j=1}^iC|r_k(t-i+j)-x_k^e(t-i+j|t-i)|.
\end{align}

Using the same process as in \eqref{e(t+rho)zhon}-\eqref{r-x} for $|r_k(t-i+j)-x_k^e(t-i+j|t-i)|$ in \eqref{ekt+rhozhon2}, by recursion and noting that $x_k^e(t'|t'')-x_k(t')=0$ for $0\leq t''<t'\leq \rho-1$, one obtains
\begin{align*}
|e_k(t+\rho)|\leq |\varpi_k(t+\rho,\cdots,t)|+\sum_{i=\rho}^{t+\rho-1}C|e_k(i)|.
\end{align*}
Similarly,
\begin{align*}
&|e_k(t+\rho-1)|\leq |\varpi_k(t+\rho-1,\cdots,t-1)|\hspace{-1mm}+\hspace{-1mm}\sum_{i=\rho}^{t+\rho-2}C|e_k(i)|,\\
&\cdots\\
&|e_k(\rho)|=|\varpi_k(\rho,\cdots,0)|.
\end{align*}
By recursion,
\begin{align}\label{evarpi}
&|e_k(t+\rho)|\leq |\varpi_k(t+\rho,\cdots,t)|+\sum_{i=0}^{t-1}C |\varpi_k(i+\rho,\cdots,i)|\nonumber\\
&\qquad\qquad\ \triangleq \hspace{-0.6mm}|\varpi_k(t+\rho,\cdots,t)|+\hspace{-0.6mm}\sum_{i=0}^{t-1}K^{i,t} |\varpi_k(i+\rho,\cdots,i)|.
\end{align}
The second claim of Lemma \ref{lemmabound} is thus proved. \hfill$\blacksquare$

\emph{Proof of Theorem \ref{Th1}}: From Lemma \ref{lemmaestimate}, one obtains
\begin{align}\label{limitamvarpi}
\lim_{k\to\infty} \frac{(a_k(t)\varpi_k(t+\rho,\cdots,t))^2}{m_k(t)^2}=0.
\end{align}

If $|\epsilon_k(t+\rho)|>\frac{\hat{w}_k(t)}{m_k(t)^2}$, then $a_k(t)=1-\frac{\hat{w}_k(t)}{|\varpi_k(t+\rho,\cdots,t)|}$ and
\begin{align}\label{varpi1}
a_k(t)|\varpi_k(t+\rho,\cdots,t)|=|\varpi_k(t+\rho,\cdots,t)|-\hat{w}_k(t).
\end{align}

If $|\epsilon_k(t+\rho)|\leq \frac{\hat{w}_k(t)}{m_k(t)^2}$, \textit{i.e.,} $\varpi_k(t+\rho,\cdots,t)\leq \hat{w}_k(t)$, then $a_k(t)=0$ and
\begin{align}\label{varpi2}
|\varpi_k(t+\rho,\cdots,t)| \leq a_k(t)|\varpi_k(t+\rho,\cdots,t)|+\hat{w}_k(t).
\end{align}

From \eqref{varpi1} and \eqref{varpi2}, one obtains
\begin{align}\label{varpibound}
|\varpi_k(t+\rho,\cdots,t)|&\leq a_k(t)|\varpi_k(t+\rho,\cdots,t)|+\hat{w}_k(t).
\end{align}

According to Lemmas \ref{lemmaestimate} and \ref{lemmabound}, one has
\begin{align}\label{lineargrowth}
m_k(t)&\leq C+\sum_{i=0}^{t-1} C |\varpi_k(i+\rho,\cdots,i)|\nonumber\\
&\leq C+\sum_{i=0}^{t-1} C a_k(i)|\varpi_k(i+\rho,\cdots,i)|\nonumber\\
&\leq C+C\max_{i\in \{0,\cdots,T-\rho\} } a_k(i)|\varpi_k(i+\rho,\cdots,i)|\nonumber\\
&\leq C+C\max_{l\leq k} a_l(\psi_l-\rho)|\varpi_k(\psi_l,\cdots,\psi_l-\rho)|,
\end{align}
where $ a_l(\psi_l-\rho)|\varpi_k(\psi_l,\cdots,\psi_l-\rho)|=\max_{i\in \{0,\cdots,T-\rho\} } a_l(i)|\varpi_l(i+\rho,\cdots,i)|$. From \eqref{limitamvarpi} and \eqref{lineargrowth}, one has
\begin{align}\label{limitKTL}
\lim_{k\to \infty}\frac{(a_k(\psi_k-\rho)\varpi_k(\psi_k,\cdots,\psi_k-\rho))^2}{(C+C\max_{l\leq k} a_l(\psi_l-\rho)|\varpi_k(\psi_l,\cdots,\psi_l-\rho)|)^2}=0.
\end{align}

According to \eqref{limitKTL} and the key technical lemma in \cite{goodwin2009adaptive} (Chapter 6, pp. 181-182), it follows that
$\lim_{k\to \infty}a_k(\psi_k-\rho)\varpi_k(\psi_k,\cdots,\psi_k-\rho)=0$,
which is equivalent to
\begin{align}\label{limitavarpi}
  \lim_{k\to \infty}a_k(t)\varpi_k(t+\rho,\cdots,t)=0,\ t\in \{0,\cdots,T-\rho\}.
\end{align}

Taking the upper limit on both sides of \eqref{varpibound} yields
\begin{align}\label{52}
 &\limsup_{k\to \infty}\varpi_k(t+\rho,\cdots,t)\nonumber \\
 &\leq  \limsup_{k\to \infty}\big(a_k(t)|\varpi_k(t+\rho,\cdots,t)|+\hat{w}_k(t)\big)\nonumber\\
 &=\lim_{k\to \infty} \hat{w}_k(t).
\end{align}

Taking the upper limit on both sides of \eqref{evarpi} yields
\begin{align}
    &\limsup_{k\to \infty} |e_k(t+\rho)| \nonumber\\
    &\leq \lim_{k\to \infty} \hat{w}_k(t)+\sum_{i=0}^{t-1} K^{i,t}\lim_{k\to \infty} \hat{w}_k(i)\nonumber\\
    &\leq \lim_{k\to \infty}\big(w+|\Tilde{w}_k(t)|+\sum_{i=0}^{t-1}K^{i,t}(w+|\Tilde{w}_k(i)|)\big)\nonumber\\
    &\leq w+\sqrt{2V_0(t)}+\sum_{i=0}^{t-1} K^{i,t}(w+\sqrt{2V_0(i)}).
\end{align}
The last inequality holds due to the fact that $\frac{1}{2}|\Tilde{w}_k(t)|^2\leq V_k(t)\leq V_0(t)$. The proof of Theorem \ref{Th1} is complete.  \hfill$\blacksquare$

\subsection{Proof of Corollary \ref{corollaryrho=1}}\label{APPendixC}

Considering
\begin{align*}
    \varpi_k(t+1,t)=x_k(t+1)-\hat{\boldsymbol\theta}_k(t)^{\top}\boldsymbol{f}(x_k(t),u_k(t)),
\end{align*}
one obtains $e_k(t+1)=\varpi_k(t+1,t)$. Using the same process used to obtain \eqref{52}, one can show that
\begin{align*}
    \limsup_{k\to \infty} |\varpi_k(t+1,t)| \leq \lim_{k\to \infty} \hat{w}_k(t) \leq w+\sqrt{2V_0(t)}.
\end{align*}
This completes the proof.\hfill$\blacksquare$

\subsection{Proof of Corollary \ref{corollarydisturbance}}\label{APPendixD}

For the disturbance-free system, under the AILC scheme \eqref{ALCLaw}, \eqref{GDPA1}, \eqref{GDPAmodify}, \eqref{GDPA5}, \eqref{estimateX1}-\eqref{estimateX3}, one can obtain \eqref{evarpi} and \eqref{limitavarpi} via the same process as in Appendix B. However, since $a_k(t)\equiv1$, \eqref{limitavarpi} becomes
$\lim_{k\to \infty}\varpi_k(t+\rho,\cdots,t)=0$, $t\in \{0,\cdots,T-\rho\}$.
Letting $k$ approach $\infty$ on the both sides of \eqref{evarpi}, one concludes the result in Corollary \ref{corollarydisturbance}.\hfill$\blacksquare$

\subsection{Proof of Lemma \ref{lemma1}} \label{APPendixE}

For any fixed $t$ and $k$, \eqref{iterativesolution} yields
\begin{align}
|\mathrm{u}_k^p(t)-\mathrm{u}_k^{p-1}(t)|&=|\mathcal{T}'(\mathrm{u}_k^{p-1}(t))-\mathcal{T}'(\mathrm{u}_k^{p-2}(t))|\nonumber\\
&\leq (1-\frac{d_0}{l'}) |\mathrm{u}_k^{p-1}(t)-\mathrm{u}_k^{p-2}(t)|\nonumber\\
&\leq (1-\frac{d_0}{l'})^{p-1}|\mathrm{u}_k^{1}(t)-\mathrm{u}_k^{0}(t)|.
\end{align}

Then, for any $p'>p$, one has
\begin{align}\label{limitukp}
&|\mathrm{u}_k^{p'}(t)-\mathrm{u}_k^p(t)|\leq \sum_{i=p+1}^{p'}|\mathrm{u}_k^{i}(t)-\mathrm{u}_k^{i-1}(t)|\nonumber\\
&\quad\leq \sum_{i=p+1}^{p'} (1-\frac{d_0}{l'})^{i-1}|\mathrm{u}_k^{1}(t)-\mathrm{u}_k^{0}(t)|\nonumber\\
&\quad\leq \frac{(1-\frac{d_0}{l'})^p[1-(1-\frac{d_0}{l'})^{p'-p}]}{1-(1-\frac{d_0}{l'})}
|\mathrm{u}_k^{1}(t)-\mathrm{u}_k^{0}(t)|.
\end{align}
Letting $p'$ approach infinity on both sides of \eqref{limitukp}, one obtains
\begin{align*}
|u_k(t)-\mathrm{u}_k^p(t)|\leq \frac{l'(1-\frac{d_0}{l'})^p}{d_0} |\mathrm{u}_k^{1}(t)-\mathrm{u}_k^{0}(t)|.
\end{align*}
Therefore, if $\mathrm{u}_k^{1}(t)=\mathrm{u}_k^{0}(t)$, then $\mathrm{u}_k^{p_o(k,t)}(t)=\mathrm{u}_k^1(t)=u_k(t)$, $|\mathrm{u}_k^{p_o(k,t)}(t)-u_k(t)|=0$. If $\mathrm{u}_k^{1}(t)\neq\mathrm{u}_k^{0}(t)$ then,
\begin{align*}
    |\mathrm{u}_k^{p_o(k,t)}(t)-u_k(t)|\leq \frac{l'(1-\frac{d_0}{l'})^{p_o(k,t)}}{d_0} |\mathrm{u}_k^{1}(t)-\mathrm{u}_k^{0}(t)|<\varepsilon.
\end{align*}
Lemma \ref{lemma1} is proven.  \hfill$\blacksquare$
\subsection{Proof of Theorem \ref{Th2}}\label{APPendixF}

Let $r_k^\ast(t+\rho)=\hat{\boldsymbol\theta}_k(t)^\top \boldsymbol{f}(\boldsymbol{X}_k^e(t),\mathrm{u}_k^p(t))$. Then, for each $t\in \{0,\cdots,T-\rho\}$,
\begin{align}\label{randrstar}
&|r_k(t+\rho)-r_k^\ast(t+\rho)|\nonumber\\
&=\big|\hat{\boldsymbol\theta}_k(t)^\top \big(\boldsymbol{f}(\boldsymbol{X}_k^e(t),u_k(t))-\boldsymbol{f}(\boldsymbol{X}_k^e(t),\mathrm{u}_k^p(t))\big)\big|\nonumber\\
&\leq L^u \|\hat{\boldsymbol\theta}_k(t)\|\cdot |u_k(t)-\mathrm{u}_k^p(t)|\nonumber\\
&\leq L^u \left(\sqrt{2V_0(t)}+\|\boldsymbol\theta(t)\|\right) \varepsilon.
\end{align}

Using the same procedure as in Appendix B, there exist positive constants $K_{\ast}^{i,t},\ i=0,1,\cdots,t-1$, independent of $k$, such that
\begin{align}\label{trackxrstar}
    &\limsup_{k\to \infty}|x_k(t+\rho)-r_k^\ast(t+\rho)|\nonumber\\
    &\leq \lim_{k\to \infty} \hat{w}_k(t)+\sum_{i=0}^{t-1} K^{i,t}_\ast\lim_{k\to \infty} \hat{w}_k(i).
\end{align}

Inequalities \eqref{randrstar} and \eqref{trackxrstar} yield
\begin{align*}
&\limsup_{k\to \infty} |e_k(t+\rho)|\nonumber\\
&\leq \limsup_{k\to \infty} \left(|x_k(t+\rho)\hspace{-0.5mm}-\hspace{-0.5mm}r_k^\ast(t+\rho)|\hspace{-0.5mm}+\hspace{-0.5mm}|r_k^\ast(t+\rho)\hspace{-0.5mm}-\hspace{-0.5mm}r_k(t+\rho)|\right)\nonumber\\
&\leq \lim_{k\to \infty} \hat{w}_k(t)+\sum_{i=0}^{t-1} K^{i,t}_\ast\lim_{k\to \infty} \hat{w}_k(i)\nonumber\\
&\quad+L^u \left(\sqrt{2V_0(t)}+\|\boldsymbol\theta(t)\|\right) \varepsilon\nonumber\\
&\leq w+\sqrt{2V_0(t)}+\sum_{i=0}^{t-1} K_\ast^{i,t}(w+\sqrt{2V_0(i)})\nonumber\\
&\quad+L^u \left(\sqrt{2V_0(t)}+\|\boldsymbol\theta(t)\|\right) \varepsilon.
\end{align*}
The proof of Theorem \ref{Th2} is complete. \hfill$\blacksquare$

\subsection{Proof of Corollary \ref{corollarynodis+appro}}\label{APPendixG}

Let $r_k^\ast(t+\rho)=\hat{\boldsymbol\theta}_k(t)^\top \boldsymbol{f}(\boldsymbol{X}_k^e(t),\mathrm{u}_k^p(t))$. Then, from Corollary \ref{corollarydisturbance}, one has
\begin{align}\label{59}
    \lim_{k\to\infty} |r_k^\ast(t+\rho)-x_k(t+\rho)|=0.
\end{align}
for all $t\in\{0,1,\cdots,T-\rho\}$. From \eqref{randrstar} and \eqref{59}, one obtains
\begin{align*}
    &\limsup_{k\to \infty} |e_k(t+\rho)|\nonumber\\
&\leq \limsup_{k\to \infty} \left(|x_k(t+\rho)\hspace{-0.5mm}-\hspace{-0.5mm}r_k^\ast(t+\rho)|\hspace{-0.5mm}+\hspace{-0.5mm}|r_k^\ast(t+\rho)\hspace{-0.5mm}-\hspace{-0.5mm}r_k(t+\rho)|\right)\nonumber\\
&\leq L^u \left(\sqrt{2V_0(t)}+\|\boldsymbol\theta(t)\|\right) \varepsilon.
\end{align*} The proof of Corollary \ref{corollarynodis+appro} is complete. \hfill$\blacksquare$

\bibliographystyle{IEEEtran}
\bibliography{Ref}

\end{document}